\documentclass[10pt,prx,onecolumn,sho
wpacs,amssymb,superscriptaddress]{revtex4-2}

\usepackage{graphicx}
\usepackage{amsmath}
\usepackage{amsfonts}
\usepackage{amsthm}
\usepackage{amssymb}
\usepackage{amsbsy}
\usepackage{wasysym}
\usepackage{bm}
\usepackage{mathrsfs}
\usepackage{color}
\usepackage{xcolor}
\usepackage{hyperref}
\usepackage{fullpage}
\usepackage{bbm}
\usepackage{physics}
\usepackage{mathtools}
\usepackage{tikz}
\usepackage{subcaption}
\usetikzlibrary{calc}
\hypersetup{
    colorlinks=true,
    linkcolor=red,        
    citecolor=blue,
    breaklinks=true,
    urlcolor=blue
}
\newcommand{\resub}[1]{{\color{black}#1}}

\newcommand{\defeq}{\mathrel{\mathop:}=}

\newtheorem{defn}{Definition}

\newtheorem{theorem}{Theorem}
\newtheorem{prop}{Proposition}
\newtheorem{corollary}{Corollary}
\newtheorem{lem}{Lemma}
\newtheorem{example}{Example}

\begin{document}

\title{Scalable tests of quantum contextuality from stabilizer-testing nonlocal games}

\author{Wanbing Zhao}
\affiliation{Department of
Physics and Astronomy, Rice University, 6100 Main Street
Houston, TX 77005, USA}
\affiliation{Applied Physics Graduate Program, Smalley-Curl Institute, Rice University, Houston, TX 77005, USA}

\author{H. W. Shawn Liew}
\affiliation{Department of Physics, National University of Singapore, Singapore 117551}
\affiliation{Centre for Quantum Technologies, National University of Singapore, 3 Science Drive 2, Singapore 117543}

\author{Wen Wei Ho}
\affiliation{Department of Physics, National University of Singapore, Singapore 117551}
\affiliation{Centre for Quantum Technologies, National University of Singapore, 3 Science Drive 2, Singapore 117543}

\author{Chunxiao Liu}
\affiliation{Department of Physics, University of California, Berkeley, California 94720, USA}
\affiliation{Université Paris-Saclay, CNRS, Laboratoire de Physique des Solides, 91405 Orsay, France}

\author{Vir B. Bulchandani}
\affiliation{Department of
Physics and Astronomy, Rice University, 6100 Main Street
Houston, TX 77005, USA}

\begin{abstract}
Soon after the dawn of quantum error correction, DiVincenzo and Peres observed that stabilizer codewords could give rise to simple proofs of quantumness via contextuality. This discovery can be recast in the language of nonlocal games: every $n$-qubit stabilizer state defines a specific ``stabilizer-testing'' $n$-player nonlocal game, which quantum players can win with probability one. If quantum players can moreover outperform all possible classical players, then the state is contextual. However, the classical values of stabilizer-testing games are largely unknown for scalable examples beyond the $n$-qubit GHZ state. We introduce several new methods for upper-bounding the classical values of these games. We first prove a general coding-theory bound for all stabilizer-testing games: if the classical value $p_{\mathrm{cl}}^* < 1$, then $p_{\mathrm{cl}}^* \leq 7/8$, i.e., there is no classical strategy that can perform as well as the optimal quantum strategy even in an asymptotic sense. We then show how to tighten this bound for the most common scalable examples, namely GHZ, toric-code and cyclic cluster states. In particular, we establish an asymptotically tight upper bound for cyclic cluster states using transfer-matrix methods. This leads to the striking conclusion that measuring an exponentially small fidelity to the cyclic cluster state will suffice to witness its contextuality. 
    
\end{abstract}

\maketitle

\section{Introduction}
 The laws of quantum physics are incompatible with a description of quantum measurement outcomes in terms of predetermined classical values~\cite{bell1966problem,kochenspecker,greenberger1989going,mermin1993hidden,Abramsky_2011}. This inherent feature of quantum mechanics, known as contextuality, 
enables quantum systems to outperform their classical counterparts at certain information-processing tasks~\cite{PhysRevA.88.022322,howard2014contextuality,Abramsky2015,Cabello21}. For example, in the early 2000s, Brassard, Broadbent and Tapp (building on earlier work by Mermin~\cite{mermin1990extreme}) showed that quantum contextuality implied the existence of a ``parity game'' that quantum players could 
 win with higher probability than any possible set of classical players, provided they were allowed to share an entangled Greenberger-Horne-Zeilinger (GHZ) state~\cite{greenberger1989going} before playing the game~\cite{brassard2004recasting}. More recently, such multi-party nonlocal games have attracted growing interest as a means 
 for rigorously demonstrating quantum computational advantage 
for constrained families of circuits~\cite{Bravyi_2018,watts2019exponential,bravyi2020quantum}, probing phases of quantum matter~\cite{Daniel_2021,bulchandani2023playing,Bulchandani_2023,Hart_2025,hart2025braiding} and experimentally benchmarking the preparation of non-trivial quantum states on noisy intermediate-scale quantum devices~\cite{Sheffer_2022,Daniel_2022,Drmota_2025,Hart_2025,kumar2025quantumclassicalseparationboundedresourcetasks}.

Given the crucial role that quantum error correction plays in quantum computing, a natural class of nonlocal games to consider from the viewpoint of benchmarking quantum devices are games whose entangled resource states are simultaneously the codewords of a quantum error-correcting code. The fact that quantum codewords could exhibit quantum contextuality was first noticed by DiVincenzo and Peres~\cite{DiVincenzo_1997} in the setting of few-qubit stabilizer codes~\cite{gottesman1997stabilizer}. This observation can be generalized to a Bell inequality for arbitrary graph states~\cite{Scarani_2005,guhne2005bell}, which admits an equivalent formulation as a nonlocal game~\cite{Cabello_2013}. In contrast to the parity game, which reflects the stabilizer algebra of the $n$-qubit GHZ state~\cite{greenberger1989going,mermin1990extreme,brassard2004recasting}, the resulting ``stabilizer-testing games'' can be defined for any stabilizer state whatsoever~\cite{zheng2024efficientverificationstabilizercode,chen2024quantumsubspaceverificationerror}. Such games have a quantum value of unity: a set of quantum players who share the appropriate stabilizer state before playing the game can \emph{always} win the game with probability one. However, the scaling of 
the optimal classical winning probabilities, or ``classical values'', of these games as the number of qubits increases remains largely unknown. (For precise definitions of quantum and classical players, see Section \ref{sec:stabtest}.)



Instead, previous studies~\cite{DiVincenzo_1997,guhne2005bell,Cabello_2013,Abramsky_2017} have mainly focused on determining whether or not stabilizer-testing games admit non-zero quantum advantage~\footnote{In this paper, we follow recent work on nonlocal games~\cite{Cabello21,lin2023quantumtasksassistedquantum,Hart_2025,Drmota_2025,kumar2025quantumclassicalseparationboundedresourcetasks} in referring to the advantage that quantum players attain at nonlocal games as ``quantum advantage'' for such games. This is a distinct concept from quantum computational advantage, which we do not address here.}, i.e., whether their classical values are strictly less than one, rather than attempting to estimate (or even upper bound) their classical values directly. This is presumably because obtaining the classical values of such games is expected to be computationally hard in general~\cite{XOR}.
Even for the best studied examples of toric code and cyclic cluster states, previous theoretical work~\cite{guhne2005bell,Cabello_2013,Bravyi_2018,Daniel_2021,Daniel_2022,Bulchandani_2023,bulchandani2023playing,Hart_2025,hart2025braiding} has tended to focus on many-body embeddings of either few-player games or parity games. Thus the classical values of these games do not clearly reflect the physics of the underlying states, being either independent of the number of qubits or identical to the classical value of a parity game played with some subset of the quantum degrees of freedom.
A better understanding of the range of possibilities for these games is desirable both from the viewpoint of benchmarking quantum state preparation for many-qubit codewords, such as surface-code states, and for clarifying the significance of quantum contextuality for quantum many-body physics, which remains relatively obscure~\cite{Daniel_2021,Bulchandani_2023,bulchandani2023playing,Hart_2025,hart2025braiding} compared to better studied measures of quantumness such as entanglement entropies~\cite{Vidal_2003,latorre2004groundstateentanglementquantum,Pasquale_Calabrese_2004,kitaev2006topological,levin2006detecting,ES}.

Motivated by such considerations, this paper develops a systematic theory of stabilizer-testing games. Our starting point is an expression for the classical values of such games in terms of classical coding theory (see also~\cite{trandafir2022irreducible}). A detailed analysis of this expression reveals that all deterministic classical strategies for the stabilizer-testing game lie in a vector subspace $V$ of the space of quadratic Boolean functions on $n$ input bits, and can thus be viewed as codewords of a classical Reed-Muller code. Meanwhile, the set of Pauli measurement outcomes for a given stabilizer state define a distinct $n$-input Boolean ``parity function'' $c(\mathbf{x})$, which can always be chosen to be at most cubic in $\mathbf{x}$ (Theorem \ref{thm:unifying bound}). The stabilizer-testing game admits non-zero quantum advantage if and only if there is a non-zero Hamming distance between $V$ and $c(\mathbf{x})$. The larger this Hamming distance, the smaller the classical value of the game. This intuitive picture allows us to develop a comprehensive theoretical understanding of nonclassicality for these games.

The paper is structured as follows. In Section \ref{sec:stabtest}, we summarize the basic properties of stabilizer-testing games, provide some explicit few-qubit examples of such games, and explain how these games connect to related constructions in the literature. We then prove (in Section \ref{sec:codingbounds}) a universal coding-theory upper bound on the classical values of stabilizer-testing games with quantum advantage, showing that the classical value $p_{\mathrm{cl}}^* \leq 7/8$ whenever the game queries a full stabilizer group. This bound follows from the code distance of the classical Reed-Muller code $\mathrm{RM}(3,n)$~\cite{macwilliams1977theory} and substantially improves upon the previous~\cite{guhne2005bell} generic bound $p_{\mathrm{cl}}^* \leq 1 - \frac{1}{2^n}$ reflecting a non-zero quantum advantage. We next explain how this $7/8$ threshold 
\resub{can be refined to an asymptotic upper bound of $3/4$ as $n \to \infty$ for GHZ states, and $1/2$ for a special family of two-dimensional translation-invariant codes (including toric-code states),} through a more detailed analysis of the nonlinearity profiles of the underlying Boolean functions. We then use transfer-matrix methods in Section \ref{sec:cluster} to establish an asymptotic upper bound of $1/2$ for cyclic cluster states. The latter result in particular implies that a fidelity-per-qubit $\epsilon \gtrsim 0.89$ to the cyclic cluster state is sufficient to demonstrate its quantum contextuality. In Section \ref{sec:MBQC}, we explain how our work connects to the established theory of contextuality in measurement-based quantum computing. We conclude by discussing various possible extensions of our results.


\section{The stabilizer-testing game}
\label{sec:stabtest}
\subsection{Overview}
The main object of study in this paper will be a specific stabilizer-testing nonlocal game that can be defined for any stabilizer state. This game was first introduced as a Bell inequality for graph states in Ref.~\cite{guhne2005bell} before being reformulated as the ``graph game'' in Ref.~\cite{Cabello_2013}. However, because our results do not make explicit use of the graph-state formalism and because this game is naturally applicable to self-testing~\cite{mayers2004selftestingquantumapparatus} of stabilizer states, we instead refer to this game as the ``stabilizer-testing game''.
    Indeed, variants of this game are implicit in recent work on verifying stabilizer codespaces~\cite{zheng2024efficientverificationstabilizercode,chen2024quantumsubspaceverificationerror}.
        
    To define the stabilizer-testing game, let $\mathcal{H}$ denote an $n$-qubit Hilbert space and let $|\psi_0\rangle \in \mathcal{H}$ be an arbitrary stabilizer codeword with stabilizer group $S = \langle g_1,....,g_{n}\rangle$. Write $(\sigma^0_j,\sigma^1_j,\sigma^2_j,\sigma^3_j) = (\mathbbm{1}_j,X_j,Y_j,Z_j)$ for the Pauli operators acting on qubit $j$. A stabilizer $M \in S$ is a signed Pauli string $M$ such that $M|\psi\rangle = |\psi\rangle$
    and can be written as 
    \begin{equation}
    \label{eq:Paulirep}
     M = (-1)^{c(M)} \prod_{j=1}^n \sigma^{\alpha_j(M)}_j
    \end{equation}
    where $c(M)$ labels the ``sign'' of the stabilizer and the indices $\alpha_j \in \{0,1,2,3\}$ specify the $j$th tensor factor of the stabilizer $M$.
    Every such state naturally defines the following nonlocal game~\cite{guhne2005bell,Cabello_2013}, that can be seen as a generalization \footnote{The Mermin--Brassard--Broadbent--Tapp parity game differs slightly from the version defined in Definition~\ref{def:stabdef}; the distinction will be clarified in Section~\ref{sec:codingbounds}.} of the Mermin-Brassard-Broadbent-Tapp parity game for the GHZ state~\cite{mermin1990extreme,brassard2004recasting}.
    
    \begin{defn}
        \label{def:stabdef}
        Let $|\psi_0\rangle$ be a stabilizer codeword with stabilizer group $S$. The \emph{stabilizer-testing game} $G(S)$ is the following $n$-player nonlocal game. At the start of the game, the referee chooses a stabilizer $M \in S$ uniformly at random and hands each player $j=1,2,\ldots,n$ the question $\alpha_j(M) \in \{0,1,2,3\}$, which specifies the $j$th Pauli operator of $M$. After this, the players are not allowed to communicate classically with one another. At the end of the game, each player $j$ must return a bit $b_j \in \{0,1\}$ to the referee. The players collectively win the game if
        
        \begin{equation}
            \sum_{j=1}^n b_j \equiv c(M) \pmod{2}.
        \end{equation} 
    \end{defn}
    Throughout this paper, we will distinguish between \emph{quantum players} and \emph{classical players} of this game, who may enact quantum strategies and classical strategies respectively. When the players of the game are quantum, they may coordinate strategies and share a globally entangled resource state before the game is played and perform gates or measurements on their own set of qubits after the game begins. When the players of the game are classical, they may coordinate strategies before the game is played, but are not allowed to share entangled resource states. Thus a \emph{quantum strategy} $\mathcal{S} =(|\psi\rangle,\mathcal{P})$ consists of a shared entangled resource state $|\psi\rangle$ and a protocol $\mathcal{P}$ that describes the set of gates and measurements that each player $j$ performs on their own set of qubits in response to the question $\alpha_j$. Meanwhile a \emph{classical strategy} consists of $n$ functions $\{b_j(\alpha_j)\}_{j=1}^n$ mapping input questions to output bits, where we consider only deterministic classical strategies without loss of generality~\cite{brassard2005quantum}. For more precise definitions, we refer the reader to Refs.~\cite{brassard2005quantum,1313847}.
    
    Every stabilizer-testing game admits the obvious perfect quantum strategy, that we call the ``default strategy'': the players share the state $|\psi_0\rangle$ before playing the game and after the game begins, each player $j$ only has access to the $j$th qubit. At the end of the game, player $j$ returns the bit $b_j$ corresponding to the outcome $(-1)^{b_j}$ of measuring the Pauli operator $\sigma_j^{\alpha_j(M)}$ specified by the question $\alpha_j(M)$ on their qubit. Let us denote this default strategy by $\mathcal{S}_0 = (|\psi_0\rangle,\mathcal{P}_0)$, where $\mathcal{P}_0$ denotes the ``default protocol'' of measurements that the players apply to their qubits.
    
    It follows that $p^*_{\mathrm{qu}}(G) = 1$, where $p^*_{\mathrm{qu}}(G)$ denotes the ``quantum value'' of $G$, i.e. the optimal probability for a set of quantum players to win the game. We would like to understand the ``classical value'' $p^*_{\mathrm{cl}}(G)$ of $G$, i.e. the optimal probability for a set of classical players to win the game. 
    
    If $p^*_{\mathrm{cl}}(G)=p^*_{\mathrm{qu}}(G) = 1$, then the stabilizer-testing game does not demonstrate quantumness of $|\psi\rangle$. If instead
    \begin{equation}
        \label{eq:quantumgame}
        p^*_{\mathrm{cl}}(G) < p^*_{\mathrm{qu}}(G) = 1,
    \end{equation}
then the stabilizer-testing game admits non-zero quantum advantage and yields a proof of quantumness for $|\psi\rangle$. 

In such cases, the stabilizer-testing game moreover furnishes a proof of quantum contextuality~\cite{bell1966problem,kochenspecker,greenberger1989going,mermin1993hidden,Abramsky_2011} for the state $|\psi_0\rangle$. This is because the subset of classical strategies that is constrained to return $b_j=0$ in response to the identity question $\alpha_j$ is in one-to-one correspondence with assignments of classical, local, hidden variables $(-1)^{b_j}$ to the outcome of measuring $\sigma_j^{\alpha_j(M)}$ in the state $|\psi_0\rangle$. Denoting the resulting ``local-hidden-variable'' value of the game by $p_{\mathrm{l.h.v.}}^*(G)$, it is clear that 
\begin{equation}
    \label{eq:LHVcomment}
    p_{\mathrm{l.h.v.}}^*(G) \leq p_{\mathrm{cl}}^*(G),
    \end{equation}
    implying that Eq. \eqref{eq:quantumgame} is a sufficient condition for contextuality of $C$. We find it simplest to focus on deriving upper bounds on $p_{\mathrm{cl}}^*(G)$ in this paper, with the understanding that this always implies an upper bound on $p_{\mathrm{l.h.v.}}^*(G)$ via Eq. \eqref{eq:LHVcomment}. In more formal language, $p_{\mathrm{l.h.v.}}^*(G) < 1$ demonstrates strong contextuality of the stabilizer subtheory~\cite{Abramsky_2011,frembs2023hierarchies}, which we will refer to simply as ``contextuality'' throughout this work.

Even though this game was introduced a long time ago~\cite{guhne2005bell,Cabello_2013}, its properties do not appear to have been studied for many-body states beyond the GHZ state, for the simple reason that $p_{\mathrm{cl}}^*(G)$ is prohibitively hard to compute in general (likely NP hard~\cite{XOR}). This means that for commonly arising stabilizer states like toric-code and cyclic cluster states, alternative games have been developed that are tailored to the specific state in question, see Refs.~\cite{Bulchandani_2023,Hart_2025,hart2025braiding} for three distinct proposals for the toric code and Refs.~\cite{Bravyi_2018,Daniel_2021} for cluster states. Possible downsides of these ``tailored'' constructions from the viewpoint of benchmarking quantum devices is that (i) they do not generalize straightforwardly to arbitrary stabilizer codewords (ii) if one applies the default protocol to a generic state $|\psi\rangle$, one generally learns little about the fidelity of the state $|\psi\rangle$ to the state $|\psi_0\rangle$ of interest.

In contrast, the probability of winning the stabilizer-testing game with the default protocol $\mathcal{P}_0$ applied to an arbitrary state $|\psi\rangle$ is
\begin{equation}
p_{\mathrm{qu}}(G,| \psi \rangle) = \frac{1}{2} + \frac{1}{2^{n+1}}\sum_{M \in S} \langle \psi | M | \psi \rangle.
\end{equation}
To express this as a fidelity, we can reintroduce the generators $g_i$ of the stabilizer group and write
\begin{equation}
\sum_{M \in S} \langle \psi | M | \psi \rangle = \sum_{\mathbf{x}\in\{0,1\}^n} \langle \psi | \prod_{j=1}^n g_j^{x_j}| \psi \rangle = \\\langle \psi | \prod_{j=1}^n (\mathbbm{1}+g_j)| \psi\rangle  = 2^n \langle \psi | \psi_0 \rangle \langle \psi_0 | \psi \rangle,
\end{equation}
which yields the concise ``fidelity formula''~\cite{Daniel_2022}
\begin{equation}
\label{eq:fidelityformula}
p_{\mathrm{qu}}(G,| \psi \rangle) = \frac{1}{2}\left(1+|\langle \psi|\psi_0\rangle|^2\right).
\end{equation}
This should be contrasted with the analogous result for the standard parity game~\cite{mermin1990extreme,brassard2004recasting} for the GHZ state, which states that~\cite{bulchandani2023playing}
\begin{equation}
\label{eq:fidelityformulaGHZ}
p_{\mathrm{qu}}(G,|\psi\rangle) = \frac{1}{2}(1+|\langle\psi|\mathrm{GHZ}^+\rangle|^2-|\langle\psi|\mathrm{GHZ}^-\rangle|^2)
\end{equation}
is penalized by the fidelity to the opposite-parity GHZ state, where $|\mathrm{GHZ}^{\pm}\rangle = \frac{1}{\sqrt{2}}(|00\ldots 0\rangle \pm |11\ldots 1\rangle)$ and we refer to $|\mathrm{GHZ}^+\rangle$ as the GHZ state throughout this paper.

We now discuss various representative examples of stabilizer-testing games and explain how they differ from related constructions in the literature.
\subsection{Examples of stabilizer-testing games}
\subsubsection{Example 1: the GHZ state}
\label{subsubsec:GHZexample}
First consider the three-qubit GHZ state $|\psi_0\rangle = \frac{1}{\sqrt{2}}(|000\rangle+|111\rangle)$ with stabilizer group
\begin{equation}
S = \langle XXX, ZZI, IZZ\rangle = \{III,IZZ,ZIZ,ZZI,XXX,-XYY,-YXY,-YYX\}.
\end{equation}
The corresponding three-player stabilizer-testing game can be expressed in terms of the following question-answer pairs
$$
\{(\boldsymbol{\alpha}(M_j),c(M_j))\}_{j=1}^{8} =\{(000, 0),\ (033, 0),\ (303,0),\ (330, 0),\ (111, 0),\ (122, 1),\ (212, 1),\ (221, 1)\},
$$
where the operators $M_j$ are the stabilizers of $|\psi_0\rangle$, the strings $\boldsymbol{\alpha}(M_j)$ encode questions for individual players and $c(M_j)$ are the stabilizer parities as defined in Eq. \eqref{eq:Paulirep}. 

As an illustrative example, consider one round of the game corresponding to the query $M_6 = -XYY$. This is encoded as the string $\mathbf{\alpha}(M_6) = 122$, which means that player one receives the question $\alpha_1 = 1$, player two receives the question $\alpha_2=2$ and player three receives the question $\alpha_3=2$, with each player ignorant of the others' questions. The corresponding parity bit $c(M_6)=1$ means that in order to win this round of the game, the players must each return bits $b_j\in\{0,1\}$ such that $\sum_{j=1}^3 b_j \equiv 1$ modulo two. Winning the game over many rounds is made classically hard by the fact that the players are not allowed to communicate with one another during each round of the game. 

This construction generalizes to an $N$-player game for all $N \geq 3$, with a classical winning probability $p_{\mathrm{cl}}^* \rightarrow 3/4^+$ as $N \rightarrow \infty$. While the precise details of the optimal classical strategy depend non-trivially on $N$~\cite{brassard2004recasting}, an asymptotically tight lower bound on $p_{\mathrm{cl}}^*$ is given by the strategy $b_j=0$ that always returns zero. Meanwhile the optimal quantum strategy (whereby player $j$ measures the Pauli matrix $\sigma_j^{\alpha_j(M_j)}$ in the state $|\psi_0\rangle$) succeeds with probability $p_{\mathrm{qu}}^* = 1$. This follows by construction: the questions that define the game correspond precisely to the stabilizers of the GHZ state.

Note that in querying \emph{all} the stabilizers of the GHZ state, our game differs from the standard formulation of the parity game~\cite{mermin1990extreme,brassard2004recasting}, according to which only stabilizers consisting of $X$ and $Y$ operators are queried.  In our three-qubit example, this restricts the set of allowed question-answer pairs to $\{(111, 0),\ (122, 1),\ (212, 1),\ (221, 1)\}$, and more generally, corresponds to querying only $2^{n-1}$ of the $2^n$ stabilizers of the $n$-qubit GHZ state. By contrast, the stabilizer-testing game for the GHZ state involves querying its full set of $2^n$ stabilizers, which is why the latter exhibits the asymptotic classical value $p_{\mathrm{cl}}^* \to 3/4^+$ as $n \to \infty$, as opposed to $p_{\mathrm{cl}}^* \to 1/2^+$ as $n \to \infty$ for the conventional parity game~\cite{brassard2004recasting}.
    
\subsubsection{Example 2: the toric code}
The stabilizer-testing game for the toric code~\cite{Kitaev_2003} is specified by the full set of local stabilizers for the toric code on a two-dimensional $d \times d$ lattice with periodic boundary conditions. The stabilizer-testing game can be extended to arbitrary genus and arbitrary surface-code terminations~\cite{bravyi1998quantumcodeslatticeboundary}, as we discuss in more detail in Section \ref{sec:codingbounds}. This stands in contrast to previous constructions of nonlocal games for toric code states~\cite{Bulchandani_2023,Hart_2025,hart2025braiding}, all of which involve querying specific subsets of the stabilizer group related to the toric code's topological order. The relative merits of each of these constructions depends on the context in which they are being applied. One appealing feature of the present construction is that it directly probes the fidelity to the toric code state via Eq. \eqref{eq:fidelityformula}. For the previous constructions, the analogous formula is either far more complicated or simply unknown. By the same token, a shortcoming of our construction is that the classical value $p_{\mathrm{cl}}^*$ is difficult to determine (though we derive upper bounds in Section \ref{sec:codingbounds} that should suffice for comparison with near-term experiments), in contrast to the previous approaches that embed the parity game in a subset of toric code stabilizers and therefore inherit the classical value of the parity game~\cite{bulchandani2023playing,Hart_2025,hart2025braiding}.

\begin{figure}[t]
        \centering
\includegraphics[width=0.75\linewidth]{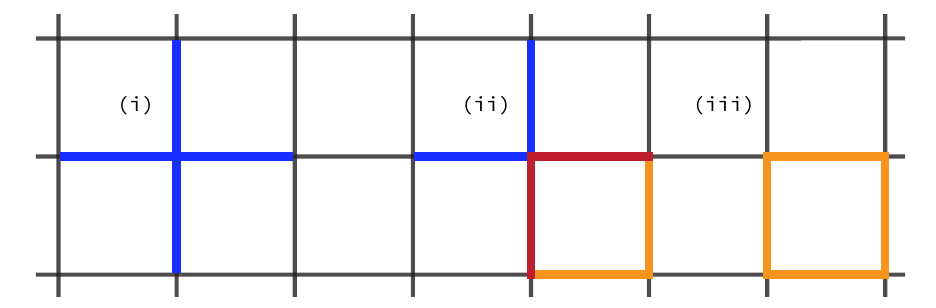}
        \caption{(i) A star operator $A_v$ of the toric code that acts on four adjacent qubits with Pauli $X$ operators (blue). (ii) The intersection of a star and plaquette operator that multiplies the Pauli representation Eq. \eqref{eq:Paulirep} by a global phase of $-1$ and generates a pair of Pauli $Y$ edges (red). (iii) A plaquette operator $B_p$ of the toric code that acts on four adjacent qubits with Pauli $Z$ operators (orange).}
        \label{fig:toric_code_stabilizers}
    \end{figure}

The generators of the stabilizer group for the toric code are~\cite{Kitaev_2003}
    \begin{equation}
         S_{\mathrm{toric}} = \big\langle\, 
        A_v = \!\!\prod_{i:\, v \in \partial i}\! X_i ,\;
        B_p = \!\!\prod_{j \in \partial p}\! Z_j
      \,\big\rangle,
    \end{equation}
    usually known as the star ($A_p$) and plaquette ($B_p$) operators, see Fig. \ref{fig:toric_code_stabilizers}. In general, stabilizers $M$ with a nonzero parity bit $c(M)$ can only be generated by intersections between star and plaquette operators (or more generally electric and magnetic Wilson lines), see Fig. \ref{fig:toric_code_stabilizers} for an illustration.

\subsubsection{Example 3: the cyclic cluster state}
Cluster states are defined by an underlying graph $H(V,E)$ where $V$ is the set of vertices and $E$ is the set of edges. Given one such graph, we can define the corresponding cluster state~\cite{BriegRauss,SchlingemannWerner}
    \begin{equation}
        \ket{\psi_H} = \Bigg( \prod_{(i,j) \in E} CZ_{i,j} \Bigg) \ket{+}^{\otimes |V|}
    \end{equation}
    The generating set of the stabilizer group of the graph state is defined by 
    \begin{equation}
        S_H=\Bigl\{\, X_i \prod_{j \in N(i)} Z_j \;\Big|\; i=1,\dots,n \,\Bigr\},
    \end{equation}
where $N(i)$ denotes the graph neighborhood of $i$. A cyclic cluster state is constructed from the cycle graph $C_n$, i.e. a ring of $n$ sites.

As an illustrative example, consider the stabilizer-testing game for the $3$-qubit cyclic cluster state, with stabilizer group
    \begin{equation}
        S = \langle XZZ,ZXZ,ZZX\rangle = \{III,XZZ,ZXZ,ZZX,IYY,YIY,YYI,-XXX\}.
    \end{equation}
The eight possible queries $M_j \in S$ yield question-answer pairs
$$
\{(\boldsymbol{\alpha}(M_j),c(M_j))\}_{j=1}^{8} =\{(000, 0),\ (133, 0),\ (313,0),\ (331, 0),\ (022, 0),\ (202, 0),\ (220, 0),\ (111,1)\}.
$$

This yields a nonlocal game with non-zero quantum advantage as in Section \ref{subsubsec:GHZexample}, with $p_{\mathrm{cl}}^* = 7/8$ and $p_{\mathrm{qu}}^* = 1$. While this game has been studied before~\cite{guhne2005bell,Cabello_2013,Daniel_2022}, the behaviour of its classical value as $n \to \infty$ was previously unknown. We provide asymptotically tight lower and upper bounds on $p_{\mathrm{cl}}^*$, together with exact values for $n \leq 16$ qubits, in Section \ref{sec:cluster} of this paper. For small numbers of qubits $n=3,4$ this classical value coincides with that of the ``triangle game'' that was previously proposed for cyclic cluster states~\cite{Bravyi_2018,Daniel_2021}. However, in contrast to the triangle game, which admits eight possible queries and has a constant classical value $p_{\mathrm{cl}}^* = 7/8$ for all even $n \geq 6$, the stabilizer-testing game for the cyclic cluster state is defined for all $n \geq 3$, admits $2^n$ possible queries and has a classical value $p_{\mathrm{cl}}^* \to 1/2^+$ as $n \to \infty$ (see Section \ref{sec:cluster}). 

The triangle game was used previously to probe $\mathbb{Z}_2\times \mathbb{Z}_2$-symmetry-protected topological (SPT) order~\cite{Daniel_2021}, but only diagnoses a subset of the SPT phase and does not uniquely specify the cyclic cluster state~\cite{bulchandani2023playing}. The stabilizer-testing game for the cyclic cluster state would be a natural starting point for remedying these shortcomings, although we will not be concerned with diagnosing phases of quantum matter in this paper.

\subsection{Linear algebraic formulation}

A convenient way to describe stabilizer-testing games is through the linear–algebra formalism introduced by Watts \emph{et al.}~\cite{XOR} for general XOR games. Stabilizer-testing games can be viewed as a special case of XOR games with each player’s input drawn from an alphabet of size~4, corresponding to the set $(\mathbbm{1},X,Y,Z)$ of possible Pauli operators at each site.

This allows stabilizer-testing games to be formulated as a constraint satisfaction problem, as follows. Labeling stabilizers \(M\in S\) as \(M_1,M_2,\ldots,M_{2^n}\), we define the \(2^n\times 4n\) \emph{incidence matrix} \(A\) by
\begin{equation}
\label{eq:defA}
A_{i,4(j-1)+k+1} =
\begin{cases}
    1, & \alpha_j(M_i) = k,\\
    0, & \text{otherwise},
\end{cases}
\end{equation}
where \(i=1,2,\ldots,2^n\) indexes stabilizers, \(j=1,2,\ldots,n\) labels players, and \(k=0,1,2,3\) enumerates Pauli operators.  
The associated \(2^n\)-component \emph{parity vector} is defined by
\begin{equation}
\label{eq:defC}
c_i = c(M_i),
\end{equation}
where $c$ is the sign of the stabilizer defined in Eq. \eqref{eq:Paulirep}. A deterministic classical strategy specifies output bits \(b_j\in\{0,1\}\) for every question \(\alpha_j\in\{0,1,2,3\}\) at site \(j\), which can equivalently be represented as a vector \(\mathbf{b}\in\{0,1\}^{4n}\).  
In terms of these quantities, the quantumness of a stabilizer-testing game can be expressed as the following purely linear-algebraic criterion:

\begin{defn}
A stabilizer-testing game instance \(G\) with incidence matrix \(A\) and parity vector \(\mathbf{c}\)
does not admit quantum advantage if there exists \(\mathbf{b}\in\{0,1\}^{4n}\) satisfying
\begin{equation}
\label{eq:classconstr}
A\mathbf{b}\equiv \mathbf{c}\pmod{2}.
\end{equation}
The instance \(G\) \emph{admits non-zero quantum advantage} if no such \(\mathbf{b}\) exists.
\end{defn}

A necessary and sufficient condition for a game \(G\) to admit non-zero quantum advantage is the existence of ``refutation'' of the system Eq. \eqref{eq:classconstr}, defined as follows~\cite{XOR}.
\begin{defn}
A \emph{refutation} of the system Eq. \eqref{eq:classconstr} is a vector $\mathbf{k} \in \mathbb{F}_2^{2^n}$ such that 
\begin{equation}
\label{eq:refutation}
\mathbf{k} \in \ker{A^\mathsf{T}} \quad \mathrm{and} \quad \mathbf{k}^{\mathsf{T}} \mathbf{c} \equiv 1 \pmod{2}.
\end{equation}
\end{defn}
This criterion is a direct consequence of the rank–nullity theorem~\cite{XOR} 
and Gaussian elimination over~$\mathbb{F}_2$ can be used to find such refutations. However,
the number of rows of $A$ grows exponentially as $n$ increases, so determining the existence of refutations in this way is not computationally efficient in our setting. In fact, for the specific case of stabilizer-testing games,
a simple sufficient condition for quantumness follows from the graph-state formalism: a vertex in the underlying graph of degree at least two implies the existence of a so-called ``all-versus-nothing triple'' of stabilizers, and thus implies quantum contextuality~\cite{Abramsky_2017}.

However, our goal in this paper is not to identify necessary and sufficient conditions for quantumness, but rather to characterize the classical value of \(G\). To this end, note that for any classical strategy \(\mathbf{b}\), we have
\begin{equation}
\label{eq:cl_prob}
p_{\mathrm{cl}}(\mathbf{b}) = 1 - \frac{d_{\mathrm{H}}(A\mathbf{b},\mathbf{c})}{2^N},
\end{equation}
where \(d_{\mathrm{H}}\) denotes the Hamming distance. The optimal classical value is therefore
\begin{equation}
\label{eq:cl_val}
p_{\mathrm{cl}}^{*}(G) = 1 - \frac{d_{\mathrm{H}}\!\bigl(\operatorname{im}(A),\mathbf{c}\bigr)}{2^n},
\end{equation}
with
\begin{equation}
d_{\mathrm{H}}\!\bigl(\operatorname{im}(A),\mathbf{c}\bigr)
:= \min_{\mathbf{b}\in\{0,1\}^{4n}} d_{\mathrm{H}}\!\bigl(A\mathbf{b},\mathbf{c}\bigr),
\end{equation}
where the image is taken over \(\mathbb{F}_2\). See Theorem 2 of Ref. \cite{trandafir2022irreducible} for an analogous coding-theory formula in the setting of quantum contextuality.

A lower bound on classical value $p_{\mathrm{cl}}^*$ 
can also be obtained straightforwardly. 
Since the trivial strategies $(0,\dots,0)$ and $(1,\dots,1)$ are contained in $\mathrm{im}(A)$, it follows that
\begin{equation} \label{eq:lower-bound}
    p_{\mathrm{cl}}^* \ge 
    \max \left\{ \frac{\mathrm{wt}(\mathbf{c})}{2^n},\, 1 - \frac{\mathrm{wt}(\mathbf{c})}{2^n} \right\}.
\end{equation}
Defining the \emph{bias} of $\mathbf{c}$ as 
$\mathrm{bias}(\mathbf{c}) \coloneqq \bigl|\frac{1}{2}-\frac{\mathrm{wt}(\mathbf{c})}{2^n}\bigr|$, Eq.~\eqref{eq:lower-bound} can be written compactly as 
$p_{\mathrm{cl}}^* \ge \tfrac{1}{2}+\mathrm{bias}(\mathbf{c})$, which indicates that a small bias in $\mathbf{c}$ is necessary for the game to exhibit a small classical value. 
In particular, we are interested in the property of ``robustness'' of scalable families of stabilizer-testing games, namely, the existence of a nonzero asymptotic difference between quantum and classical values in the limit of a large number of players $n \gg 1$. (An example of a game that fails to be robust in this sense is the polygon game introduced in Ref.~\cite{bulchandani2023playing}.)
Following the terminology of Ref.~\cite{XOR}, we define the \emph{asymptotic difference} for a scalable family of stabilizer-testing games \(\{G_n\}_{n\in\mathbb{N}}\), where \(G_n\) denotes the game instance with system size \(n\), as follows: 

\begin{defn}[Asymptotic difference]
\label{def:asymdiff}
\begin{equation}
\Delta = 2\!\left[1 - \lim_{n\to\infty} p_{\mathrm{cl}}^*(G_n)\right].
\end{equation}
\end{defn}
Equation~\eqref{eq:lower-bound} implies that $\Delta \le 1$ 
always. Achieving an ``asymptotically perfect difference''~\cite{XOR} $\Delta = 1$ means that the optimal classical strategy performs no better than random guessing as $n \to \infty$. Meanwhile an asymptotic difference $\Delta =0$ implies that the games $G_n$ cease to admit non-zero quantum advantage as $n \to \infty$.

We note that there are two standard~\cite{briet2013explicit,XOR} ways of measuring the advantage attainable by quantum players at such games, namely the ``bias difference'' $2(p_{\mathrm{qu}}^* - p_{\mathrm{cl}}^*)$, implicit in Definition~\ref{def:asymdiff}, and the ``bias ratio''~\cite{XOR} or quantum-classical gap~\cite{briet2013explicit} $(p_{\mathrm{qu}}^* - \tfrac{1}{2})/(p_{\mathrm{cl}}^* - \tfrac{1}{2})$. When the quantum value $p_{\mathrm{qu}}^*$ is replaced by an experimentally measured quantum winning probability $p_{\mathrm{qu}}(|\psi\rangle)$, e.g. as in Eq. \eqref{eq:fidelityformula}, which measure of quantum advantage is more practically useful depends on the fidelity to the ideal resource state that can be attained. For high fidelities corresponding to relatively small numbers of qubits or state-of-the-art quantum devices, the bias difference should yield a clearer measure of quantum advantage for nonlocal games~\cite{briet2013explicit,bulchandani2023playing,Hart_2025}. However, for realistic systems of many qubits, the bias difference will generally be exponentially small in $n$~\cite{bulchandani2023playing}, suggesting that the bias ratio is a more natural measure of such ``weak'' quantum advantage~\cite{lin2023quantumtasksassistedquantum}. Our primary goal in this paper is to obtain upper bounds on $p_{\mathrm{cl}}^*$ for various examples of interest, which can equivalently be formulated as lower bounds on both the bias difference and the bias ratio.


The 
bias ratio diverges as $n \to \infty$ for the parity game and for the stabilizer-testing game for the cyclic cluster state. This should be contrasted with two-player XOR games, for which the Tsirelson bound $(p_{\mathrm{qu}}^* - \tfrac{1}{2})/(p_{\mathrm{cl}}^* - \tfrac{1}{2}) \le K_G^{\mathbb{R}}$ holds, where $K_G^{\mathbb{R}} \lesssim 1.78$ denotes the real Grothendieck constant~\cite{tsirel1987quantum,briet2013explicit}.


\section{Coding-theory upper bounds on the classical value}
\label{sec:codingbounds}
In this paper, we present various complementary approaches for upper bounding the classical values of stabilizer-testing games, both of which exploit the Boolean functions (or, equivalently, the Reed–Muller codes) induced by stabilizer states. Previous studies have shown that the nonlinearity profile of Boolean functions can be used to estimate non-classical resources in certain computational models. For example, the magic (non-stabilizerness) of hypergraph states can be bounded by nonquadraticity (second-order nonlinearity) of their Boolean function representations~\cite{liu2022many}. In the setting of non-adaptive stabilizer measurement-based quantum computation (MBQC), it has similarly been established that the success probability of computing a given Boolean function is bounded above by its nonquadraticity~\cite{PhysRevA.88.022322,frembs2023hierarchies}. Here we establish another such connection between Boolean functions and quantumness of states. Specifically, we show how to upper bound the classical values of stabilizer-testing games from the nonlinearity and nonquadraticity of certain Boolean functions. We connect these results to the established theory of contextuality in MBQC~\cite{PhysRevA.88.022322,frembs2023hierarchies} in Section~\ref{sec:MBQC}.

In this section, we first show that stabilizer-testing games as in Definition~\ref{def:stabdef} that admit non-zero quantum advantage can be formulated as approximating a cubic polynomial (corresponding to quantum measurement outcomes) by a family of quadratic polynomials (corresponding to classical strategies). Since this approximation must always be imperfect, a unifying Reed-Muller upper bound on its classical value follows directly (see Theorem \ref{thm:unifying bound}). We then analyze the asymptotic behavior of these games as $n \to \infty$. We point out that the conventional parity game~\cite{brassard2004recasting} relies only on the nonlinearity of the GHZ parity function, which does not exhibit substantial nonquadraticity. In contrast, the induced parity functions of cyclic cluster states (derived in Section~\ref{sec:cluster}) exhibit nonquadraticity that is maximal in the large-system limit. In this sense, the stabilizers of the cyclic cluster state are ``more contextual'' than those for the GHZ state. Finally, we provide an upper bound on the classical value of the stabilizer-testing game defined by the toric code.

\subsection{General bound from Reed-Muller codes}
In the interests of generality, it will be helpful to extend our definition of the stabilizer-testing game, Definition \ref{def:stabdef}, to arbitrary stabilizer codespaces, as follows. Although our results could be extended to qudit systems, we will focus on qubits for the remainder of this paper, as the most standard and experimentally relevant setting.
\begin{defn}[$(\mathcal{C}_n,\mathcal{M})$ stabilizer testing game]\label{def:stabilizer game}
    Let $\mathcal{C}_n \subseteq (\mathbb{C}^2)^{\otimes n}$ be an $n$-qubit stabilizer code with stabilizer group
    \[
        S \;=\; \{\, M \in \mathcal{P}_n \;\mid\; M \ket{\psi} = \ket{\psi},\ \forall \ket{\psi} \in \mathcal{C}_n \,\}.
    \]
    Let $\mathcal{M} = \{M_1,M_2,\ldots,M_l\} \subseteq S$ be a fixed subset. In a $(\mathcal{C}_n,\mathcal{M})$ stabilizer-testing game, the referee chooses $M_i \in \mathcal{M}$ uniformly at random and hands each player $j=1,2,\ldots,n$ the question
    \[
        q_i^j \;=\; \pi\!\big(\mathbf{v}_{j,n+j}(M_i)\big),
    \]
    where $\mathbf{v} = (\mathbf{x}\,|\,\mathbf{z}) \in \mathbb{F}_2^{2n}$ is the symplectic representation of $M_i$, $\mathbf{v}_{j,n+j} := (v_j, v_{n+j})$ denotes the pair of bits corresponding to the $j$-th and $(n+j)$-th positions of $\mathbf{v}$, and $\pi:\mathbb{F}_2^2 \to \{0,1,2,3\}$ is the standard Pauli encoding. Allowing no communication between players, each player~$j$ returns a bit $b_j \in \{0,1\}$ to the referee. The players collectively win the game if
        
        \begin{equation}
            \sum_{j=1}^n b_j \equiv c(M_i) \pmod{2},
        \end{equation} 
 where $c(M_i) \in \{0,1\}$ denotes the parity bit of the stabilizer $M_i$ defined in Eq. \eqref{eq:Paulirep}.
\end{defn}

A stabilizer code $\mathcal{C}_n$ can be specified by a maximal set $\mathcal{G} = \{g_1, g_2, \dots, g_r\}$ of $r\leq n$ independent elements in $S$, such that $S = \langle \mathcal{G} \rangle$, i.e., any element $M \in S$ can be written as
\begin{equation}\label{eq:stab_op}
    M = g_1^{x_1} g_2^{x_2} \cdots g_r^{x_r},
\end{equation}
where each $x_i \in \{0,1\}$. For a quantum code specified by such $r$ generators, there are $2^{2^r}$  subsets $\mathcal{M}$ of $S$, each giving rise to a distinct stabilizer-testing game. In particular, when $\mathcal{M} = S$, the vectors in the column span of the incidence matrix $A$, together with the parity vector $\mathbf{c}$, defined by analogy with Eqs. \eqref{eq:defA} and \eqref{eq:defC} respectively,
are binary vectors of length $2^r$ and can therefore be viewed as classical codewords in $\mathbb{F}_2^{2^r}$, and Eq. \eqref{eq:cl_val} is replaced by
\begin{equation}
\label{eq:cl_val_r}
p_{\mathrm{cl}}^{*}(G) = 1 - \frac{d_{\mathrm{H}}\!\bigl(\operatorname{im}(A),\mathbf{c}\bigr)}{2^r}.
\end{equation}
in this more general setting.

Specifically, the $(\mathcal{C}_n,S)$ stabilizer-testing game has two properties that make it amenable to analysis. First, once a generating set $\mathcal{G}$ of $\mathcal{C}_n$ is fixed, the vectors in the column span of $A$ and the parity vector $\mathbf{c}$ can be represented as Boolean functions of the $r$ binary variables $\mathbf{x}=(x_1,\dots,x_r)$. Thus, $\operatorname{im}(A)$ and $\mathbf{c}$ can be embedded into the Reed--Muller code $\mathrm{RM}(s,r)$, defined as the set of all binary vectors $\mathbf{f}$ obtained by evaluating Boolean functions 
$f(x_1,\dots,x_r)$ that are polynomials of degree at most $s$ over $\mathbb{F}_2$. We will refer to the Boolean function $c(\mathbf{x})$ defined by $\mathbf{c}$ in this way as the \emph{parity function} associated with a given generating set. Second, for stabilizer-testing games, these Boolean functions 
can be chosen to be low-degree functions. The minimum Hamming distance of the order $s$ Reed--Muller code $RM(s,r)$ is known to be 
\[
d_{\min}(RM(s,r)) = 2^{\,r-s}.
\]
By Eq.~\eqref{eq:cl_val_r}, it follows that $p_{\mathrm{cl}}^* \leq 1-2^{-s}$. If we knew solely that the game admitted non-zero quantum advantage, this would yield the trivial bound $p_{\mathrm{cl}}^* \leq 1 - 2^{-r}$, which would not be useful for comparison with experiments in many-body regimes $r \gg 1$. However, for stabilizer-testing games that query the full stabilizer group, we can always constrain $s \leq 3$, giving rise to a universal and experimentally accessible bound $p_{\mathrm{cl}}^* \leq 7/8$ for any $(\mathcal{C}_n,S)$ stabilizer-testing game. The formal statement is as follows.

\begin{theorem}\label{thm:unifying bound}
    If a $(\mathcal{C}_n,S)$ stabilizer-testing game admits non-zero quantum advantage, then its classical value is at most \(7/8\). 
\end{theorem}
\begin{proof}
    Let $(\mathcal{C}_n,S)$ be a stabilizer-testing game with incidence matrix $A$ and parity vector $\mathbf{c}$, viewed as Boolean functions of $\mathbf{x}$. Suppose that $\operatorname{im}(A)$ and $\mathbf{c}$ are spanned by polynomials of degree at most $s$. Then the vector space $\mathbf{c} \oplus \operatorname{im}(A)$ is a subspace of $RM(s,r)$. Because the game admits non-zero quantum advantage, $\mathbf{c} \notin \mathrm{im}(A)$ and therefore $d_{\mathrm{H}}(A\mathbf{b},\mathbf{c} )\geq 2^{r-s}$ for all classical strategies $\mathbf{b}$. By Proposition \ref{prop:cubic}, we can always choose a generating set of $S$ such that $s \leq 3$ and the result follows by Eq. \eqref{eq:cl_val_r}.
\end{proof}

We will see that the degree of the column-space vectors of $A$ depends on the local Hilbert space dimension, $\dim(\mathcal{H})$, which in our case equals~2. Moreover, the degree of the parity vector $\mathbf{c}$ is bounded, up to a constant factor, by that of the column-space vectors of $A$. 
To see this, we analyze how the signs of Pauli operators transform, borrowing a useful formula from the derivation of the quadratically signed weight enumerator~\cite{knill2001quantum}.

\begin{lem}\emph{\cite{knill2001quantum}}
\label{lem:weight enu}
    For any $\mathbf{v}_1,\mathbf{v}_2 \in \{0,1\}^{2n}$, let $\sigma_{\mathbf{v}_1}, \sigma_{\mathbf{v}_2}$ be the corresponding Pauli operators in symplectic representation. Then
    \[
        \sigma_{\mathbf{v}_1+\mathbf{v}_2} = e^{i\phi}\,\sigma_{\mathbf{v}_1}\sigma_{\mathbf{v}_2},
    \]
    where the phase satisfies
    \begin{equation} \label{eq:sign}
        \frac{\phi}{\pi} \;=\; \tfrac{1}{2}\Big(W_Y(\mathbf{v}_1)+W_Y(\mathbf{v}_2)-W_Y(\mathbf{v}_1+\mathbf{v}_2)\Big) \;+\; \mathbf{v}_1^{\mathsf T} B \mathbf{v}_2,
    \end{equation}
    with $W_Y(\mathbf{v})$ denoting the number of $Y$ operators in $\sigma_\mathbf{v}$, and
    \[
    B \;=\;
    \begin{bmatrix}
        0 & I_n \\
        0 & 0
    \end{bmatrix}
    \]
    the $2n\times 2n$ block matrix with $I_n$ the $n\times n$ identity. 
\end{lem}

\begin{proof}
   First replace Pauli $Y$ operators with $iY$. Let $\tilde{\sigma}_\mathbf{v}$ denote the Pauli operators after this substitution, $\tilde{\sigma}_\mathbf{v} = i^{W_Y(\mathbf{v})}\sigma_\mathbf{v}$. 
   Then one can verify directly that $\tilde{\sigma}_{\mathbf{v}_1+\mathbf{v}_2}=(-1)^{\mathbf{v}_1^{\mathsf{T}} B \mathbf{v}_2}\tilde{\sigma}
_{\mathbf{v}_1}\tilde{\sigma}_{\mathbf{v}_2}$.
\end{proof}

Using this relation, we can prove that the degree of the parity function for $(\mathcal{C}_n,S)$ 
stabilizer-testing games can always be chosen to be at most three, as follows. 
\begin{prop}
\label{prop:cubic}
For every $(\mathcal{C}_n,S)$ stabilizer testing game, there exists a generating set $\mathcal{G}$ of $S$ such that the game admits non-zero quantum advantage if and only if its parity function is a cubic polynomial.
\end{prop}
\begin{proof}
By induction on Lemma \ref{lem:weight enu}, we have $\sigma_{\mathbf{u}_1+\dots+\mathbf{u}_r} = e^{i\phi}\sigma_{\mathbf{u}_1}\dots \sigma_{\mathbf{u}_r}$, where
\begin{align}
    \frac{\phi}{\pi} 
    &= \tfrac{1}{2} \left( \sum_{i=1}^r W_Y(\mathbf{u}_i) - W_Y\!\left(\oplus_{i=1}^r \mathbf{u}_i\right) \right)
       + \sum_{1 \leq i < j \leq r} \mathbf{u}_i^{\mathsf T} B \mathbf{u}_j.
\end{align}
    Given a generating set $\mathcal{G} = \{g_1, g_2, \dots, g_r\}$ for 
    $\mathcal{C}_n$, now let $\mathbf{v}_i \in \mathbb{F}_2^{2n}$ be the symplectic binary vector for generator $g_i$ and $s_i\in \mathbb{F}_2$ be the parity of $g_i$. Then, letting $\mathbf{u}_i = x_i\mathbf{v}_i$ in the above expression, the parity function can be written as the following polynomial in $\mathbf{x} = (x_1, \dots, x_r)$:
\begin{equation}\label{eq:parity function}
    c(\mathbf{x}) \;\equiv\; 
    \tfrac{1}{2} \left( \sum_{i=1}^r x_i W_Y(\mathbf{v}_i) 
    - W_Y(\mathbf{v}_{\mathbf{x}}) \right)
    + \sum_{1 \leq i < j \leq r} \mathbf{v}_i^{\mathsf T} B \mathbf{v}_j \, x_i x_j
    + \sum_{i=1}^r s_i x_i 
    \pmod{2},
\end{equation}
      where $\mathbf{v}_{\mathbf{x}}=\oplus_{i=1}^r x_i\mathbf{v}_i$ is the symplectic representation of the operator $M(\mathbf{x}) = \prod_{i=1}^{r} g_i^{x_i}$.


Next we show that the first term in Eq.~\eqref{eq:parity function} is of degree at most 3. 
Since the generating set $\mathcal{G}$ for a given stabilizer group $S$ is not unique, we are free to choose a generating set that contains only two types of non-identity Pauli matrices at each site. 
By Lemma \ref{lem: Gaussian elimination} of Appendix \ref{app:proofs}, such a $\mathcal{G}$ always exists. 

Our strategy will be to use invariance of the classical value under local Clifford operations to express the resulting incidence matrix in a tractable form. In the notation of Eq. \eqref{eq:Paulirep}, any state $|\psi\rangle \in \mathcal{C}_n$ satisfies the relations
\begin{equation}
\prod_{j=1}^n \sigma_j^{\alpha_j(M(\mathbf{x}))} |\psi\rangle = (-1)^{c(\mathbf{x})} | \psi\rangle, \quad \mathbf{x} \in \{0,1\}^r.
\end{equation}
Acting on both sides of this equation by a product of local Clifford operations $U = \prod_{i=1}^n U_i$ yields a different stabilizer state $|\psi'\rangle = U|\psi\rangle$, for which 
\begin{equation}
\prod_{j=1}^n \left(U_j\sigma_j^{\alpha_j(M(\mathbf{x}))} U_j^\dagger\right)|\psi'\rangle = (-1)^{c(\mathbf{x})} | \psi'\rangle, \quad \mathbf{x} \in \{0,1\}^r.
\end{equation}
Whenever conjugation by $U_j$ yields a simple permutation of the Pauli matrices, this induces a new $(U\mathcal{C}_n,USU^\dagger)$ stabilizer testing game, with the same parity function $c'(\mathbf{x})=c(\mathbf{x})$ as the $(\mathcal{C}_n,S)$ stabilizer testing game and a new incidence matrix $A'$ that is column equivalent to $A$. We therefore restrict each $U_j$ to lie in the subgroup  of Clifford operations $\langle R_{2\pi/3}\rangle \cong C_3$ corresponding to cyclic permutations of Pauli operators
$R_{2\pi/3} = (1-iX-iY-iZ)/2$~\cite{Grier_2022}. By our previous choice of $\mathcal{G}$, there exist local Cliffords $U_j \in \langle R_{2\pi/3}\rangle$ such that all rows of $A'$ corresponding to the generating set $\mathcal{G}' = U\mathcal{G}U^\dagger$ are only supported on Pauli $Z$ and $X$ operators.  

After performing this redefinition of the stabilizer-testing game, we drop the prime and work with the redefined generators and incidence matrix. Writing the latter as \(A=[\mathbf{a}_1,\dots,\mathbf{a}_{4n}]\), let $\mathcal{X}, \mathcal{Y}, \mathcal{Z}, \mathcal{I} \subseteq [4n]$ 
denote the disjoint index sets of columns in~$A$ corresponding to Pauli $X, Y, Z,$ and~$I$, respectively, so that $\mathcal{X} \,\dot\cup\, \mathcal{Y}\,\dot\cup\, \mathcal{Z} \,\dot\cup\, \mathcal{I} = [4n]$. For each index set \( P \in \{\mathcal{X}, \mathcal{Y}, \mathcal{Z}, \mathcal{I}\} \) corresponding to Pauli \( \mathcal{P} \in \{X, Y, Z, I\} \), let \( P = \{p_1, \dots, p_n\} \) (where $p$ stands for $x$, $y$, $z$, or $\iota$, respectively) and denote by $\mathbf{a}_{p_k}$ the $k$th column of $A$ corresponding to the Pauli $\mathcal{P}$. For each site $j \in [n]$, let $\mathcal{P}_j \subseteq [r]$ be the set of generator indices $i$ for which the $j$th site of generator $g_i$ acts as Pauli~$\mathcal{P}$.

We have $W_Y(\mathbf{v}_{\mathbf{x}}) = \sum_{j\in \mathcal{Y}} a_j(\mathbf{x})$. Let $J := \{\, j\in [n] \mid X_j \neq \varnothing \text{ and } Z_j \neq \varnothing \,\}$ be the set of sites $j$ for which both $X_j$ and $Z_j$ are nonempty. Then by our choice of generating set
\begin{equation}\label{eq:incidence-polynomial}
    a_{y_j}(\mathbf{x)} =
    \begin{cases}
        \bigl(\displaystyle\bigoplus_{i \in X_j} x_i\bigr)
        \bigl(\displaystyle\bigoplus_{i \in Z_j} x_i\bigr), & j \in J, \\[4pt]
        0, & j \in [n] \setminus J.
    \end{cases}
\end{equation}
The corresponding contribution to the overall parity is then
\begin{align}
    \tfrac{1}{2} \sum_{j \in \mathcal{Y}} a_j(\mathbf{x})
    &\equiv \tfrac{1}{2} \sum_{j \in J}
        \Biggl(
            \sum_{k \in X_j} x_k
            - 2 \sum_{\substack{k,l \in X_j \\ k < l}} x_k x_l
        \Biggr)
        \Biggl(
            \sum_{k \in Z_j} x_k
            - 2 \sum_{\substack{k,l \in Z_j \\ k < l}} x_k x_l
        \Biggr) \\
    &\equiv \tfrac{1}{2} \sum_{j \in J} \sum_{\substack{k \in X_j \\ l \in Z_j}}
         x_k x_l
        + \sum_{j \in J}
        \Biggl(
            \sum_{\substack{k,l \in X_j \\ k < l,\; m \in Z_j}}
                x_k x_l x_m
            + \sum_{\substack{k,l \in Z_j \\ k < l,\; m \in X_j}}
                x_k x_l x_m
        \Biggr) \pmod{2}. \label{eq:cubic-parity-function}
\end{align}

We note that for any two generators $k \neq l$, the number of summands in the first term of this expression involving $x_kx_l$ is precisely $\mathbf{v}_k^{\mathsf{T}} B \mathbf{v}_l + \mathbf{v}_l^{\mathsf{T}} B \mathbf{v}_k$, by definition of $B$, while terms with $k=l$ do not appear by our choice of $\mathcal{G}$. Thus $\tfrac{1}{2} \sum_{j \in J} \sum_{k \in X_j, l \in Z_j} x_k x_l = \frac{1}{2}\sum_{1\leq k < l \leq r}(\mathbf{v}_k^{\mathsf{T}} B \mathbf{v}_l + \mathbf{v}_l^{\mathsf{T}} B \mathbf{v}_k) x_k x_l$. Moreover, since the generators pairwise commute, we must have $\mathbf{v}_k^{\mathsf{T}} B \mathbf{v}_l + \mathbf{v}_l^{\mathsf{T}} B \mathbf{v}_k \equiv 0 \pmod{2}$. It follows that each $x_k x_l$ with $k \neq l$ appearing in the first term has an integer coefficient, and therefore that
\begin{equation}
\label{eq:simplified-parity-function}
c(\mathbf{x}) \equiv \sum_{j \in J}
        \Biggl(
            \sum_{\substack{k,l \in X_j \\ k < l,\; m \in Z_j}}
                x_k x_l x_m
            + \sum_{\substack{k,l \in Z_j \\ k < l,\; m \in X_j}}
                x_k x_l x_m
        \Biggr)  + \frac{1}{2}\sum_{1\leq i<j \leq r} (\mathbf{v}_i^{\mathsf{T}} B \mathbf{v}_j -\mathbf{v}_j^{\mathsf{T}} B \mathbf{v}_i) x_i x_j + \sum_{i=1}^r s_i x_i \pmod{2},
\end{equation}
which is a cubic polynomial iff there exists $j\in J$ such that either $X_j$ or $Z_j$ have two or more elements. 

For each column type $P \in \{\mathcal{X},\mathcal{Z},\mathcal{I}\}$ 
and each Pauli label $\mathcal{P} \in \{X,Z,I\}$, we have
\begin{equation}\label{eq:incidence-polynomial-P}
  a_{p_j} (\mathbf{x}) \;=\;
  \begin{cases}
    \displaystyle\bigoplus_{i \in \mathcal{P}_j} x_i, & j \in J_{\mathcal{P}}, \\[4pt]
    0, & j \in [n] \setminus J_\mathcal{P} ,
  \end{cases}
\end{equation}
where $J_{\mathcal{P}} \subseteq [n]$ is the set of sites on which 
at least one generator acts as~$\mathcal{P} \in \{X,Z,I\}$. Thus, by Eqs.~\eqref{eq:incidence-polynomial} and~\eqref{eq:incidence-polynomial-P}, 
each column vector $\mathbf{a}_j$ of the incidence matrix~$A$ 
corresponds to a Boolean polynomial in the variables $\mathbf{x}$ of degree at most~$2$. This proves the ``if'' part of the theorem statement. 

For the ``only if'' part, suppose that $c(\mathbf{x})$ is not cubic. Then for each $j \in J$, $X_j$ and $Z_j$ have at most one element each, and by Eq. \eqref{eq:simplified-parity-function},
\begin{equation}\label{eq:parity-function-noAVN}
c(\mathbf{x}) \equiv \frac{1}{2}\sum_{1\leq i<j \leq r} (\mathbf{v}_i^{\mathsf{T}} B \mathbf{v}_j -\mathbf{v}_j^{\mathsf{T}} B \mathbf{v}_i) x_i x_j + \sum_{i=1}^r s_i x_i \pmod{2}.
\end{equation}
We observe that each nonzero case of Eqs.~\eqref{eq:incidence-polynomial} and ~\eqref{eq:incidence-polynomial-P} with $\mathcal{P} \in \{X,Z\}$ is a monomial. The sets \(
\{\, a_{y_j}(\mathbf{x}) : j\in J \,\}
\) and \(
\{\, a_{x_j}(\mathbf{x}) : j\in J_X \,\}
\cup
\{\, a_{z_j}(\mathbf{x}) : j\in J_Z \,\}
\)
therefore span the subspaces of quadratic and linear terms that can appear in Eq.~\eqref{eq:parity-function-noAVN}.
Thus, $\mathbf{c}\in\mathrm{im}(A)$. In particular, $n\geq 3$ is a necessary condition for non-zero quantum advantage in this setting~\cite{Abramsky_2017}.

\end{proof}

We calculate the parity function for graph state, GHZ state and toric code as examples.

\begin{example}[Graph states \& Graph codes]
    Let $H = (V,E)$ be a graph with adjacency matrix $\Gamma$, where $V$ is the set of $n$ vertices and $E$ the set of edges.   
    The associated \emph{graph state} $\mathcal{C}_n \subseteq \mathcal{H}_2^{\otimes n}$ is defined as the codespace stabilized by
    \[
        \mathcal{G}_H=\Bigl\{\, X_i \prod_{j \in N(i)} Z_j \;\Big|\; i=1,\dots,n \,\Bigr\},
    \]
    where $N(i)$ denotes the graph neighborhood of vertex $i$. This generating set has the form in Proposition \ref{prop:cubic} with parity bits $s_i=0$ for all $i$. It is well known that any stabilizer code is locally Clifford equivalent to such a graph code~\cite{schlingemann2001stabilizercodesrealizedgraph}. We have
\[
    a_{y_i}(\mathbf{x}) \;=\; x_i \bigoplus_{j \in N(i)} x_j,
    \qquad 
    \mathbf{v}_i^{\mathsf T} B \mathbf{v}_j \;=\; \Gamma_{i,j}.
\]
The quadratic term in Eq.~\eqref{eq:cubic-parity-function} coincides with the second term of Eq.~\eqref{eq:parity function}, which we denote by $\Gamma(x) = \sum_{1 \leq i < j \leq n} \Gamma_{i,j} x_i x_j$. 
Note that this quadratic form $\Gamma(x)$ is precisely the quadratic polynomial 
that defines the graph state~\cite{danielsen2005self}, given by
\begin{equation}
    \label{eq:graph-state-phase}
    |H\rangle = \frac{1}{2^{n/2}}
    \sum_{x \in \{0,1\}^n}
    (-1)^{\Gamma(x)} |x\rangle.
\end{equation}
By symmetry of the adjacency matrix $\Gamma_{i,j}=\Gamma_{j,i}$, this quadratic polynomial cancels in Eq. \eqref{eq:simplified-parity-function}, and the parity function takes the form
\begin{equation}
    \label{eq:parity_graph}
    c_{\mathrm{H}}(\mathbf{x}) 
    \;\equiv\;  \sum_{i \in V} \;\sum_{\substack{j<k \\ j,k \in N(i)}} 
        x_i x_j x_k \pmod 2.
\end{equation}
Each term in this sum captures cubic interactions among triples formed by a vertex and two of its neighbors.

Eq.~\eqref{eq:parity_graph} coincides with the polynomial $f_H$ given in Lemma~12 of Ref.~\cite{Cabello_2013}, whose Hamming weight defines the LC-invariant $\beta(H)$ that distinguishes inequivalent contextuality orbits. This correspondence is immediate in our framework.

In addition, the necessary and sufficient condition for Eq.~\eqref{eq:parity_graph} 
to yield a cubic polynomial coincides with the all-versus-nothing (AvN) triple condition proposed for graph states
in Ref.~\cite{Abramsky_2017}, namely the existence of a vertex with degree $\geq 2$. The authors of Ref.~\cite{Abramsky_2017} further conjectured 
that this condition holds for graph codes in general. This follows directly from Proposition \ref{prop:cubic}. 



\end{example}

\begin{example}[GHZ state]
     We choose the generators of n-qubit GHZ state to be \[G_{\mathrm{GHZ}} \;=\; \{\, X_1 X_2 \cdots X_n,\; Z_1 Z_2,\; Z_2 Z_3,\;\dots,\; Z_{n-1} Z_n \,\}. \] Then for $(\mathcal{C}_{\mathrm{GHZ}},S_{\mathrm{GHZ}})$ game, we have $J = [n]$ and 
     \begin{subequations}
     \begin{align}
a_{y_1} &= x_1 x_2, \label{eq:incidence GHZ_2_full}\\
a_{y_n} &= x_1 x_n, \label{eq:incidence GHZ_n_full}\\
a_{y_j} &= x_1 (x_j \oplus x_{j+1}) \quad \text{for } 1< j < n. \label{eq:incidence GHZ_j_full}
     \end{align}
     \end{subequations}
The parity bits for these generators $s_i=0$, while the non-zero contributions to Eq. \eqref{eq:simplified-parity-function} are cubic terms of the form $x_1 x_jx_{j+1}$ for $2 \leq j \leq n-1$ and quadratic terms resulting from the relation
     \begin{equation}
    \mathbf{v}_i^{\mathsf{T}} B \mathbf{v}_j = \begin{cases}
        2 & i = 1 < j \leq n, \\
        0 & \mathrm{otherwise}. 
    \end{cases}
     \end{equation}

It follows by Eq \eqref{eq:simplified-parity-function} that 
\begin{equation}\label{eq:parity GHZ_full}
    c_{\mathrm{GHZ}}(\mathbf{x}) \equiv \sum_{j=2}^{n-1} x_1 x_jx_{j+1} +\sum_{j=2}^n x_1 x_j \pmod 2.
\end{equation} 

This should be contrasted with the standard parity game for the GHZ state~\cite{mermin1990extreme,brassard2004recasting}, for which the query set is
\begin{equation}\label{eq:queries-GHZ}
    \mathcal{M}_{\mathrm{GHZ}} = X^{\otimes n}\,
\big\langle Z_1 Z_2,\; Z_2 Z_3,\;\dots,\; Z_{n-1} Z_n \big\rangle.
\end{equation} 
This is recovered from the full stabilizer-testing game for the GHZ state by restriction to the affine hyperplane $x_1 = 1$. On this hyperplane, the polynomials $a_j'$ and $c'$ become functions only of $x_2,\dots, x_n$:
\begin{subequations}\label{eq:GHZ_incidence}
\begin{align}
a_{y_1}'(\mathbf{x}) &= x_2, \label{eq:incidence GHZ_2}\\
a_{y_n}'(\mathbf{x}) &= x_n, \label{eq:incidence GHZ_n}\\
a_{y_j}'(\mathbf{x}) &= x_j \oplus x_{j+1} \quad \text{for } 1 < j < n, \label{eq:incidence GHZ_j}\\
c_{\mathrm{GHZ}}'(\mathbf{x}) 
&\equiv \sum_{j=2}^{n-1} x_j x_{j+1} 
   + \sum_{j=2}^n x_j \pmod 2. \label{eq:parity GHZ_partial}
\end{align}
\end{subequations}

\end{example}
By Eqs.~\eqref{eq:incidence GHZ_2}--\eqref{eq:incidence GHZ_j}, 
the image of the incidence map $\mathrm{im}(A_{\mathrm{GHZ}}')$, 
spans the entire linear space of $n-1$ variables. We discuss how to recover the usual classical value of the parity game~\cite{mermin1990extreme,brassard2004recasting} for all $n$ from our formalism in Proposition \ref{prop:cl-value-nl1}. However, for odd $n$ there is an appealing interpretation of the classical value in terms of bent functions~\cite{rothaus1976bent,PhysRevA.88.022322}. Specifically, the first term in Eq.~\eqref{eq:parity GHZ_partial} corresponds to a bent function when $n$ is odd, for which we obtain
\(
d_{\mathrm{H}}\!\bigl(\operatorname{im}(A_{\mathrm{GHZ}}'),\mathbf{c}_{\mathrm{GHZ}}'\bigr) = \mathrm{nl}_1(c_{\mathrm{GHZ}}') = 2^{\,n-2} - 2^{\,(n-1)/2-1},
\)
 where $\mathrm{nl}_1(\cdot)$ denotes the first-order nonlinearity (see Section~\ref{sec:Asymptotic Bound}). Since $r=n-1$ in this case, it follows by Eq. \eqref{eq:cl_val_r} that $p_{\mathrm{cl}}^*(\mathrm{GHZ})'
    \;=\; \tfrac{1}{2} + 2^{-\tfrac{n+1}{2}} = \tfrac{1}{2} + 2^{-\lceil n/2\rceil} $ for odd $n$, as usual~\cite{mermin1990extreme,brassard2004recasting}.

\begin{example}[2D toric code]

The stabilizer group for the two-dimensional toric code on an $L \times L$ torus is
\begin{equation}\label{eq:toric-generator}
    S_{\mathrm{toric}}
    = \big\langle\, 
        A_s = \!\!\prod_{i \in s}\! X_i ,\;
        B_p = \!\!\prod_{j \in \partial p}\! Z_j
      \,\big\rangle ,
\end{equation}
where $A_s$ and $B_p$ are the star and plaquette operators, respectively. The stabilizer generators satisfy the global constraints
\begin{equation}
    \prod_s A_s = \mathbb{I}, \qquad \prod_p B_p = \mathbb{I}.
\end{equation}
Consequently, only $2L^2 - 2$ of the generators are independent, yielding a stabilizer code with parameters $[[n,k,d]] = [[L^2, 2, L]]$. For the $(\mathcal{C}_{\mathrm{toric}}(L), S_{\mathrm{toric}})$ game, 
we have $J = [2L^2]$. 
Each edge $j \in J$ is incident to two distinct stars $s_j$ and $s_j'$ and two plaquettes $p_j$ and $p_j'$. The corresponding incidence function is
\[
    a_{y_j}(\mathbf{x}) = (x_{s_j} \oplus x_{s_j'}) (x_{p_j} \oplus x_{p_j'}) ,
    \qquad 
    j \in s_j \cap s_j',\;\;
    j \in \partial p_j \cap \partial p_j' .
\]
The parity function, up to global redundancies, is

\begin{equation}\label{eq:parity_toric}
    c_{\mathrm{toric}}(\mathbf{x}) \equiv \sum_{1\leq j \leq 2L^2} (x_{s_j} + x_{s_j'})x_{p_j} x_{p_j'}  + (x_{p_j} + x_{p_j'})x_{s_j} x_{s_j'} + \sum_{s_j} \sum_{\{p_k :\; s_j \cap p_k \neq \emptyset\}} x_{s_j}x_{p_k} \pmod 2.
\end{equation}
This follows by similar considerations to the previous examples, such as the relations
\begin{equation}
\mathbf{v}^{\mathsf{T}}_{s_j} B \mathbf{v}_{p_k} = \begin{cases}
2 & s_j \cap p_k \neq \emptyset \\ 
0 & \mathrm{otherwise}
\end{cases}, \quad \mathbf{v}^{\mathsf{T}}_{p_k} B \mathbf{v}_{s_j} = 0, \quad \forall j,k.
\end{equation}
Since there are $2L^2 - 2$ independent generators, 
one may fix an arbitrary star variable and an arbitrary plaquette variable to zero 
to obtain the true parity function. More
generally, the above derivation applies to surface codes on any compact, connected, orientable 2D manifold, regardless of its genus $g$. Using the Euler--Poincar\'{e} relation
\[
    V - E + F \;=\; \chi(M) \;=\; b_{0} - b_{1} + b_{2},
\]
 where \(V\), \(E\), and \(F\) denote the numbers of vertices, edges, and faces in a cellular decomposition of the surface \(M\), $\chi(M)$ denotes the Euler characteristic, and
\(b_{k} = \dim H_{k}(M;\mathbb Z_{2})\) are the corresponding Betti numbers \((b_{0}=1, b_{1}=2g, b_{2}=1)\), the number of independent generators is
\[
    E - 2g 
    \;=\; F + V - (b_{0}+b_{2}),
\]
where the subtraction of $b_{0}+b_{2}=2$ accounts for the two global 
stabilizer redundancies present on any closed surface.

When boundaries are introduced, the Boolean variables associated with the 
corresponding boundary type, either direct (plaquettes) or dual (stars), must 
be set to zero. For example, for a surface code with parameters $[[2L^2-L,1,L]]$, defined on a cylinder, removing 
$L$ faces or $L$ vertices amounts to fixing the associated variables to zero; 
one then fixes an arbitrary variable of the remaining type to eliminate the 
global redundancy. For the $[[2L(L-1)+1,1,L]]$ planar code, one similarly fixes the $L$ removed plaquette variables and the $L$ removed vertex variables to zero.

We note that this formulation of the stabilizer-testing game for the toric code only refers to local stabilizers. The corresponding version of Eq. \eqref{eq:fidelityformula} therefore keeps track of the fidelity to the codespace. By including global stabilizers, i.e. those corresponding to non-trivial cycles on a given manifold, this game generalizes easily to specific codewords, and in such cases, Eq. \eqref{eq:fidelityformula} will directly capture the fidelity to the codeword in question. Our approach therefore provides yet a fourth alternative to previous proposals for toric-code games involving global~\cite{Bulchandani_2023}, local~\cite{Hart_2025}, and ``mixed'' sets of toric-code stabilizers~\cite{hart2025braiding}, which can be seen as a unification of the previous proposals, insofar as these can all be viewed as special cases or ensembles of stabilizer-testing games.
\end{example}
    

\subsection{\resub{Asymptotic upper bounds on the classical value}}
\label{sec:Asymptotic Bound}

We now introduce better upper bounds than
Theorem~\ref{thm:unifying bound},
formulated in terms of the nonlinearity profile of the parity function. 

\begin{defn}[Nonlinearity]
Let $f : \mathbb{F}_2^r \to \mathbb{F}_2$ be a Boolean function.
Its $s$-th-order nonlinearity is defined by
\begin{equation}
    \mathrm{nl}_s(f)
    \;\defeq\;
    \min_{x \in \mathrm{RM}(s,r)}
    d_{\mathrm{H}}(x,f),
\end{equation}
where $\mathrm{RM}(s,r)$ denotes the Reed--Muller code of order~$s$
and $d_{\mathrm{H}}(\cdot,\cdot)$ is the Hamming distance.
\end{defn}

Since the Reed--Muller codes form a nested family,
\[
\mathrm{RM}(1,r)
   \subset
   \mathrm{RM}(2,r)
   \subset
   \cdots
   \subset
   \mathrm{RM}(r,r),
\]
their associated nonlinearities satisfy
\[
\mathrm{nl}_1(f)
   \;\ge\;
   \mathrm{nl}_2(f)
   \;\ge\;
   \cdots
   \;\ge\;
   \mathrm{nl}_r(f).
\]
For fixed $s$ and large~$r$, an asymptotic upper bound on
$\mathrm{nl}_s(f)$ was obtained in Ref.~\cite{carlet2006improving,carlet2008recursive}:
\begin{equation}\label{eq:upper-bound-of-nl}
    \mathrm{nl}_s(f)
    \;=\;
    2^{\,r-1}
    - \frac{\sqrt{15}}{2}\,
      (1+\sqrt{2})^{\,s-2}\,
      2^{\,r/2}
    + \mathcal{O}\!\bigl(r^{\,s-2}\bigr).
\end{equation}
We now connect such nonlinearity profiles to the classical winning probabilities of stabilizer games.

When all stabilizer elements are queried, the parity function 
$c \in \mathrm{RM}(3,r)$ is approximated by 
$\mathrm{im}(A) \subseteq \mathrm{RM}(2,r)$. In this setting, non-classicality of the game corresponds to nonquadraticity of $c$. We will call this the ``full-query regime''. When the queries are restricted to a coset of the stabilizer group obtained by fixing $t$ variables to specific binary values, various possibilities arise. If the restricted parity function $c' \in \mathrm{RM}(3,r-t)$ continues to be cubic, the analysis proceeds as for the full-query regime. If $c' \in \mathrm{RM}(2,r-t)$ is quadratic, it may contain quadratic terms of either
identical or distinct Pauli types at each site. However, the image $\mathrm{im}(A') \subseteq \mathrm{RM}(2,r-t)$ contains only quadratic terms corresponding to distinct Pauli types by Eq.~\eqref{eq:incidence-polynomial}. Thus, $\mathrm{im}(A')$ and $c'$ are quadratic but linearly independent in general. Applying the argument used in the proof of Theorem~\ref{thm:unifying bound} with $s=2$ yields
$p_{\mathrm{cl}}^* \le 3/4$. Finally, when the queries are restricted to a coset of the stabilizer group 
such that $\mathrm{im}(A') \subseteq \mathrm{RM}(1,r-t)$, the corresponding parity function 
$c' \in \mathrm{RM}(2,r-t)$ is approximated by affine functions. In this setting, non-classicality of the game corresponds to nonlinearity of $c'$. We will refer to this latter case, which includes the standard parity game~\cite{mermin1990extreme,brassard2004recasting}, as the ``partial-query regime''.


Using Eq.~\eqref{eq:cl_val}, we obtain the following bounds on 
the classical winning probability. 
For the full-query regime,
\begin{subequations}\label{eq:bound-by-nl}
\begin{equation}
    p_{\mathrm{cl}}^* \;\le\; 1 - 2^{-r}\, \mathrm{nl}_2(c),\label{eq:bound-by-nl2}
\end{equation}
and for the partial-query regime,
\begin{equation}
    p_{\mathrm{cl}}^* \;\le\; 1 - 2^{-r+t}\, \mathrm{nl}_1(c').
    \label{eq:bound-by-nl1}
\end{equation}
\end{subequations}

From Eq.~\eqref{eq:upper-bound-of-nl}, we see that the nonlinearity-based upper bounds 
permit parity functions for which $\Delta = 1$ is attainable within this framework. When using the bounds Eq. \eqref{eq:bound-by-nl2} and \eqref{eq:bound-by-nl1}, terms of degree $\leq2$ in $c$ and degree $\leq1$ in $c'$ can be omitted by definition of the nonlinearities. Furthermore, the classical value for the partial-query regime 
can be computed explicitly using the Walsh--Hadamard transform, for which the upper bound Eq. \eqref{eq:bound-by-nl1} turns out to be an equality.

\begin{prop}[Classical value via $\mathrm{nl}_1$] \label{prop:cl-value-nl1}
Let $(\mathcal{C}_n,\mathcal{M})$ be a stabilizer-testing game.  
If $\mathcal{M}$ is a coset of $S(\mathcal{C}_n)$ such that the column span of its associated incidence matrix  $A'$ is linear and its parity function $c'$ is quadratic,  
then the game satisfies
\begin{equation}\label{eq:bound-nl1-rank}
    p_{\mathrm{cl}}^* \;=\; \tfrac{1}{2} + 2^{-\tfrac{1}{2}\,\mathrm{rank}(c') - 1},
\end{equation}
where $\mathrm{rank}(c')$ denotes the rank of the associated quadratic form.
\end{prop}

\begin{proof}
Write
\(
    \phi(\mathbf{x}) = \bigl(a_1(\mathbf{x}),\dots,a_{4n}(\mathbf{x})\bigr)
\)
for the affine map $\phi:\mathbb{F}_2^{r-t} \to \mathbb{F}_2^{4n}$ corresponding to evaluation of each row of $A'$, which is injective by Definition~\ref{def:stabilizer game}. Writing $a_j(\mathbf{x}) = \mathbf{u}_j \cdot \mathbf{x} \oplus v_j$, it follows by Lemma~\ref{lem:full-rank} of Appendix~\ref{app:proofs} that the $\mathbf{u}_j$ span $\mathbb{F}_2^{r-t}$. Since $\sum_{k=1}^4 a_{4j+k}(\mathbf{x}) = 1$ by definition of the incidence matrix, we also have $1 \in \mathrm{im}(A')$. Hence $\mathrm{im}(A') = \mathrm{RM}(1,r-t)$ and by Eq. \eqref{eq:cl_val_r}, 
\begin{equation}
\label{eq:eq_from_nl1}
p_{\mathrm{cl}}^* = 1 - \frac{\mathrm{nl}_1(c')}{2^{r-t}},
\end{equation}
i.e. the upper bound in Eq. \eqref{eq:bound-by-nl1} is saturated. It remains to determine the first-order nonlinearity of $c'$.

For any Boolean function $f : \mathbb{F}_2^r \to \mathbb{F}_2$, 
its first-order nonlinearity satisfies
\begin{equation}\label{eq:nl1-formula}
    \mathrm{nl}_1(f)
    = 2^{\,r-1}
      - \tfrac{1}{2}
        \max_{u \in \mathbb{F}_2^r}
        |W_f(u)| ,
\end{equation}
where 
\resub{\begin{equation}\label{eq:walsh-spectrum}
    W_f(u)
    = \sum_{x \in \mathbb{F}_2^r}
      (-1)^{f(x)+ \langle u, x \rangle}
\end{equation}}
denotes the Walsh--Hadamard transform of $f$. For a quadratic function $f$, one further has the identity~\cite{macwilliams1977theory}
\begin{equation}\label{eq:walsh-quadratic}
    \max_{u \in \mathbb{F}_2^r} |W_f(u)|
    = 2^{\,r - \mathrm{rank}(f)/2},
\end{equation}
where $\mathrm{rank}(f)$ is the rank of the associated quadratic form. Substituting $f = c'$ and $r \mapsto r-t$ into 
Eqs.~\eqref{eq:nl1-formula}–\eqref{eq:walsh-quadratic} and using Eq. \eqref{eq:eq_from_nl1} yields the desired result.
\end{proof}
Since the rank of any quadratic Boolean polynomial over $\mathbb{F}_2$ is even, the parity function $c'$ satisfies
\(
2 \le \mathrm{rank}(c') \le n - 1,
\)
where the upper bound is strict when $n$ is even. Thus for stabilizer-testing games satisfying the hypotheses of Proposition \ref{prop:cl-value-nl1}, Eq.~\eqref{eq:bound-nl1-rank} further implies that
\begin{equation}\label{eq:pcl*-bound-quadratic}
    \tfrac{1}{2} + 2^{-\lceil n/2 \rceil} \le p_{\mathrm{cl}}^* \le \tfrac{3}{4}.
\end{equation}
To approach the lower bound in Eq.~\eqref{eq:pcl*-bound-quadratic}, we seek an \emph{extensive} rank, i.e., $\operatorname{rank}(c')$ growing linearly with the number of qubits $n$. According to Eq.~\eqref{eq:cubic-parity-function}, such an extensive rank can be realized in games defined from translation-invariant, many-body stabilizer states with spatially local stabilizer generators.
Translation invariance leads to a parity function with $\Theta(n)$ quadratic terms. With the additional assumption of spatially local stabilizer generators, the resulting parity function can be chosen to exhibit only finite-range correlations, in the sense that each Boolean variable $x_j$ appears in only $O(n^0)$ quadratic monomials, thereby giving rise to an extensive rank.

For example, the GHZ state achieves the lower bound in Eq.~\eqref{eq:pcl*-bound-quadratic} for such games by fixing the Boolean variable corresponding to the global generator $X \otimes \cdots \otimes X$ to be $1$, so that the quadratic parity function involves only 
nearest-neighbor correlations. Specifically, when $\mathcal{M}$ is chosen as in Eq.~\eqref{eq:queries-GHZ}, Proposition \ref{prop:cl-value-nl1} recovers the standard result $p_{\mathrm{cl}}^* = \tfrac{1}{2} + 2^{-\lceil n/2 \rceil}$ from Refs.~\cite{mermin1990extreme,brassard2004recasting}. 
Consequently, one obtains an asymptotically perfect difference $\Delta_{\mathrm{GHZ}}' = 1$. However, when the set of allowed measurements is taken to be the full stabilizer group, these considerations from the theory of quadratic Boolean functions no longer apply directly. Nevertheless, the direct-sum structure of Eq. \eqref{eq:classconstr} conditioned on $x_1$ immediately yields $p_{\mathrm{cl}}^*(\mathrm{GHZ}) = \frac{1}{2}(1+p_{\mathrm{cl}}^*(\mathrm{GHZ})' = \frac{3}{4} + 2^{-\lceil n/2 \rceil-1}$ in this case, implying that $\Delta_{\mathrm{GHZ}} = \tfrac{1}{2}$.


In the next section we show that even when the full set of stabilizers 
$\mathcal{M} = S_{\mathrm{cluster}}$ is used, 
the cyclic cluster state gives rise to a family of nonlocal games with asymptotically perfect difference $\Delta_{\mathrm{cluster}} = 1$. This reflects the fact that the parity function for the cyclic cluster state exhibits a greater degree of nonquadraticity than that for the GHZ state. This should also be contrasted with the conventional parity game, whose parity function exhibits nonlinearity but no nonquadraticity.


We now derive a bound on $(\mathcal{C}_n,S)$ stabilizer-testing games from the nonquadraticity of the parity function.

\begin{prop}[Bound via $\mathrm{nl}_2$]\label{prop:bound via derivative}
Let $(\mathcal{C}_n,S)$ be a stabilizer-testing game with parity function $c$. If the game admits non-zero quantum advantage, then its classical value satisfies the following two bounds:
\begin{subequations}\label{eq:derivative-bounds}
    \begin{align}
    \text{(i)}\quad
    p_{\mathrm{cl}}^*
    &\le
    \tfrac{3}{4} 
    + 2^{-\tfrac{1}{2}\,\max_{\mathbf{a} \in \mathbb{F}_2^r}\,\mathrm{rank}(D_{\mathbf{a}} c)-2},
    \label{eq:derivative-bound1}\\[4pt]
    \text{(ii)}\quad
    p_{\mathrm{cl}}^*
    &\le
    \tfrac{1}{2}
    + \tfrac{1}{2}\,
    \sqrt{\tfrac{1}{2^r}\sum_{\mathbf{a} \in \mathbb{F}_2^r}
    \!2^{-\tfrac{1}{2}\,\mathrm{rank}(D_{\mathbf{a}} c)}},
    \label{eq:derivative-bound2}
\end{align}
\end{subequations}
where $D_{\mathbf{a}} c(\mathbf{x}) = c(\mathbf{x}) \oplus c(\mathbf{x}+\mathbf{a})$ denotes the Boolean derivative of $c$ in direction $\mathbf{a} \in \mathbb{F}_2^r$, and $\mathrm{rank}(D_{\mathbf{a}} c)$ denotes the rank of the associated quadratic form.
\end{prop}

\begin{proof}
    The result follows directly from 
    Propositions~2 and~3 in Ref.~\cite{carlet2008recursive} 
    by taking $s=2$; see Appendix~\ref{app:proofs} for details.
\end{proof}

\resub{We illustrate the power of these results by using them to upper bound classical values for the 2D toric code and the cyclic cluster state using bounds (i) and (ii) respectively, see Appendix~\ref{app:examples}. Of course, for specific cases of interest, one can always compute classical values directly from Eq. \ref{eq:cl_val_r} or obtain tighter bounds than (i) and (ii) from Eq. \ref{eq:bound-by-nl2}. The utility of the results (i) and (ii) instead resides in their generality and in the remainder of this section, we focus on the general conclusions implied by these bounds. First, }these two results are complementary. For a $(\mathcal{C}_n, S)$ game that admits non-zero quantum advantage, Eq.~\eqref{eq:derivative-bound1} implies that the existence of a \emph{single} direction 
along which the derivative of the parity function has rank growing with the system size 
is sufficient to yield $\Delta = \tfrac{1}{2}$. Eq.~\eqref{eq:derivative-bound2} further shows that a rank whose uniform average over \emph{all} possible directional derivatives grows
with system size constitutes a sufficient condition for 
$\Delta = 1$. Next, we show that a coset game with maximal asymptotic difference $\Delta = 1$ defined on some state implies a lower bound on the asymptotic difference of the corresponding full-query game on the same state. \resub{(Note that this is precisely the relation between the usual parity game~\cite{mermin1990extreme,brassard2004recasting} and the stabilizer-testing game for the GHZ state discussed in this paper.)}

\begin{corollary}[] \label{col:corollary1}
If there exists a coset \(
\mathcal{M} \subseteq S(\mathcal{C}_n)\) for which the game $(\mathcal{C}_n,\mathcal{M})$, whose  incidence matrix $A'$ has a linear column span and whose parity function $c'$ is quadratic, attains maximal asymptotic difference, i.e., $\Delta(\mathcal{C}_n,\mathcal{M}) = 1$, then the asymptotic difference for game $(\mathcal{C}_n,S)$ satisfies $  \Delta(\mathcal{C}_n,S) \geq \tfrac{1}{2}$.
\end{corollary}

\begin{proof}
Assume there exists a coset 
$\mathcal{M} \subseteq S(\mathcal{C}_n)$ 
for which the game $(\mathcal{C}_n,\mathcal{M})$ 
achieves maximal asymptotic difference.  
The coset partitions the generator set $G$ into three subsets:  
(i) $k$ generators $\{g_1,\dots,g_k\}$ whose corresponding Boolean variables $\{x_1,\dots,x_k\}$ are fixed to $1$,  
(ii) $t-k$ generators $\{g_{k+1},\dots,g_{t}\}$ whose corresponding Boolean variables $\{x_{k+1},\dots,x_t\}$ are fixed to $0$, and  
(iii) $r-t$ generators $\{g_{t+1},\dots,g_r\}$ whose corresponding Boolean variables $\{x_{t+1},\dots,x_r\}$ remain free. Then by Eq.~\eqref{eq:bound-nl1-rank}, the rank of the associated parity function $c'$ diverges as $n \to \infty$. Setting $a_i = 1$ for $i = 1,\dots,k$ and $a_i = 0$ otherwise, the derivative \(D_{\mathbf{a}} c\) of the parity function \(c\) along direction \({\mathbf{a}}\)
yields a quadratic form whose coefficient matrix contains that of \(c'\) as a principal submatrix. To see this, suppose that a cubic monomial of the form
\(x_i x_j \cdot \bigoplus_{l\in L_{ij}} x_l\) appears in $c$, where
$t+1 \le i,j \le r$ and $L_{ij} \subseteq \{1,\dots,t\}$. Then the quadratic monomial $x_ix_j$ appears in the derivative $D_{\mathbf{a}} c$ if and only if $x_ix_j$ appears in $c'$, i.e., $\bigoplus_{l\in L_{ij}} x_l \equiv 1 \pmod 2$ under the restriction \(x_1=\cdots=x_k=1\) and \(x_{k+1}=\cdots=x_t=0\). Consequently, $\mathrm{rank}(D_{\mathbf{a}} c)$ diverges as $n \to \infty$.  
Applying Eq.~\eqref{eq:derivative-bound1} then yields 
$\Delta(\mathcal{C}_n,S) \ge \tfrac{1}{2}$.
\end{proof}

\resub{Thus, any nonlocal game that encodes the parity game must exhibit an asymptotic difference of $\Delta = \frac{1}{2}$ in the full-query setting (see Ref.~\cite{hart2025braiding} for examples). Finally, we identify a family of translation-invariant Calderbank-Shor-Steane (CSS) codes, which is broad enough to include the toric code, 
for which bound (ii) above guarantees that full-query games achieve an asymptotically perfect difference $\Delta = 1$. This demonstrates in particular that the phenomenon of an asymptotically perfect difference is not limited to finely tuned constructions, as one might surmise from previous work~\cite{briet2013explicit,XOR} and can instead arise in well-studied families of quantum codes.
}

\resub{Let us therefore consider translation-invariant CSS stabilizer codes on a two-dimensional square lattice.  We follow Haah's ``polynomial method'' for representing such codes~\cite{Haah_2013,haah2016algebraic}, as follows. First recall that for a translation-invariant code, the generating set for its stabilizer group can be taken to be translation invariant and consist of local Pauli operators. We take the spatial lattice $\Lambda$ to be a discrete torus on $L_y \times L_z$ sites, which can equivalently be viewed as a product of cyclic groups with generators $y$ and $z$ such that $y^{L_y}=z^{L_z}=1$. We further take each unit cell to consist of two edges, with one qubit per edge, as shown in Fig. \ref{fig:sublattice}. Then a two-qubit Pauli operator associated with a given lattice site $(a,b) \in \Lambda$, with symplectic representation $\mathbf{v} \in \mathbb{F}_2^4$, can be expressed as $P(a,b;\mathbf{v}) = y^a z^b \mathbf{v}$, and multiplication of Pauli operators at different sites can be represented as addition between such objects, which have the structure of an $R$-module over the ring $R = \mathbb{F}_2[y^{\pm 1},z^{\pm 1}]$ of formal Laurent polynomials in the generators of the lattice~\cite{haah2016algebraic}. Thus, for example, the $X$ (star) operator of the toric code centered at site $(0,0)$ would be represented as the four-component column vector $(1+y,1+z | 0,0)^{\mathsf{T}}$, while the $Z$ (plaquette) operator with top-left corner at $(0,0)$ would be represented as $(0,0|1+z^{-1},1+y^{-1})^{\mathsf{T}}$. We denote by $\overline{(\cdot)}$ the involution on $R$ that inverts the lattice generators, defined by $\overline{g(y,z)} := g(y^{-1}, z^{-1})$ for all $g \in R$. 

We now restrict our attention to translation-invariant CSS codes with one $X$-type and $Z$-type generator at each lattice site, which we represent as four-component vectors of Laurent polynomials $A_s$ and $B_p$ respectively. The ``excitation map'' of a single-qubit Pauli error $\mathcal{P}_i \in \{X,Y,Z,I\}$ on edge $e_i$ of unit cell is given by $\epsilon(\mathcal{P}_i)= [A_s \cdot \mathcal{P}_i, B_p \cdot \mathcal{P}_i]$, where $\cdot$ is the symplectic product of two $R$-module elements $p$ and $q$, defined by $p \cdot q = \overline{p}^{\mathsf{T}} \lambda q$, and $\lambda = \begin{pmatrix} 0_2 & -\mathbbm{1}_2 \\ \mathbbm{1}_2 & 0_2 \end{pmatrix}$ is the appropriate symplectic form. (In particular, a non-zero element of $\epsilon(\mathcal{P}_i)$ indicates that $\mathcal{P}_i$ does not preserve the codespace associated with the corresponding stabilizer generator.) In what follows, it will also be useful to define the ``support'' of a Laurent polynomial $h=\sum_{\boldsymbol{\gamma}\in \Lambda}h_{\boldsymbol{\gamma}} x^{\boldsymbol{\gamma}} \in R$ as 
    $\operatorname{supp}(h):=
    \{\boldsymbol{\gamma}\in \Lambda:\ h_{\boldsymbol{\gamma}}\neq 0\}$ and its ``diameter'' as $d(h):=
\max_{\boldsymbol{\gamma},\boldsymbol{\gamma}' \in \operatorname{supp}(h)}
\|\boldsymbol{\gamma}-\boldsymbol{\gamma}'\|_2 $. We then have the following theorem.}

\begin{theorem}\label{thm:2d translation invariant}
\resub{Consider the full-query game associated with a two-dimensional
translation-invariant CSS stabilizer code on the $L_y \times L_z$ torus $\Lambda$, whose generators take the form
\begin{equation} \label{eq:ldpc}
    A_{(0,0)} =
\begin{pmatrix}
f(y, z) \\
g(y, z) \\
\hline
0 \\
0 \\
\end{pmatrix},
\quad
B_{(0,0)} =
\begin{pmatrix}
0 \\
0 \\
\hline
\overline{g(y,z)} \\
\overline{f(y,z)}
\end{pmatrix}.
\end{equation}
We assume that $f,g\in \mathbb F_2[y^{\pm1},z^{\pm1}]$ have non-empty support contained in $[-D,D]^2\cap \mathbb Z^2$ for some constant $8D\leq L_y,L_z$ that is otherwise independent of the system size. We make two additional technical assumptions:
\begin{enumerate}
\item there exist two sites $\boldsymbol{\alpha}^* =(\alpha^*_1,\alpha^*_2)$ and $\boldsymbol{\alpha} =(\alpha_1,\alpha_2)$ in $\operatorname{supp}(f)$ such that $\|\boldsymbol{\alpha}-\boldsymbol{\alpha}^*\|_2=d(f)$ and $\Delta := |\alpha_2-\alpha^*_2|>0$,
\item for all pairs of distinct sites $\boldsymbol{\gamma},\tilde{\boldsymbol{\gamma}} \in \operatorname{supp}(g)$, $\tilde{\boldsymbol{\gamma}}-\boldsymbol{\gamma}$ is never parallel to $\boldsymbol{\alpha}^*-\boldsymbol{\alpha}$ as a vector in $\mathbb{R}^2$.
\end{enumerate}
For this class of 2D stabilizer codes, 
\begin{equation}
p_{\mathrm{cl}}^*
\leq
\frac12+\frac12
\left(\cos\frac{\pi}{8}\right)^{
\frac{L_yL_z}{4m\,d(f)^2}},
\end{equation}
where $m=\left\lceil
\frac{2\max_{\boldsymbol{\gamma}\in\operatorname{supp}(g)}
\|\boldsymbol{\gamma}\|_2}
{\Delta}
\right\rceil + 1$.}
\end{theorem}
\begin{proof}
\resub{Motivated by Eq.~\eqref{eq:derivative-bound2}, we seek a principal submatrix of the matrix that defines the quadratic form \(D_{\mathbf a}c\), such that \(\operatorname{rank}(D_{\mathbf a}c)\) can be
lower bounded by the rank of this submatrix, which is easier to compute. The CSS codes we consider exhibit a natural bipartition of their stabilizer generators into
\(X\)-type generators $A_{(\alpha,\beta)} = y^{\alpha}z^{\beta} A_{(0,0)}$ and \(Z\)-type generators $B_{(\alpha,\beta)} = y^{\alpha}z^{\beta} B_{(0,0)}$. This bipartition simplifies the calculation of the parity function. In particular, we will construct a tractable principal submatrix that is induced only by \(X\)-type generators. This will first require ``coarse-graining'' the lattice according to the locality of the generators, as follows (see Fig.~\ref{fig:sublattice} for an illustration). The idea will be to sparsify the set of $X$-type generators such that surviving $X$-type generators share at most a single qubit. The $Z$-type generators that anti-commute with this ``dangling'' Pauli $X$ may then be enumerated exhaustively, which will allow us to quantify the cubic nonlinearity associated with these qubits using Eq.~\eqref{eq:simplified-parity-function}. 

To proceed further, let $\boldsymbol{\alpha}^*
=(\alpha^*_1,\alpha^*_2)
,\boldsymbol{\alpha}=(\alpha_1,\alpha_2)$ be two sites in $\operatorname{supp}(f)$ that are farthest apart, i.e.,
$\|\boldsymbol{\alpha}-\boldsymbol{\alpha}^*\|_2=d(f)$, as in the theorem statement.
Then $\operatorname{supp}(f)$ is confined to a lens-shaped region between $\boldsymbol{\alpha}^*$ and $\boldsymbol{\alpha}$, as illustrated in Fig.~\ref{fig:sublattice}. Thus $\|\boldsymbol{\alpha}-\boldsymbol{\alpha}^*\|_2$ can be thought of as the interaction diameter of $f$. Let $\boldsymbol{\beta}$ be a site in $\operatorname{supp}(g)$ farthest from the origin of the lattice, i.e., $\boldsymbol{\beta}
\in
\operatorname*{arg\,max}_{\boldsymbol{\gamma}\in\operatorname{supp}(g)}
\|\boldsymbol{\gamma}\|_2$, where $\|\boldsymbol{\beta}\|_2$ can be thought of as the interaction radius of $g$. By assumption, $\alpha_2\neq \alpha_2^*$, which implies that the vector $\mathbf{y}' = (\alpha_1-\alpha
_1^*,\alpha_2-\alpha_2^*)$ and its orthogonal vector $\mathbf{z}' = m(\alpha_2-\alpha_2^*,\alpha_1^*-\alpha
_1)$, where $m=\left\lceil
\frac{2\|\boldsymbol{\beta}\|_2}
{\left|\alpha_2-\alpha_2^*\right|}
\right\rceil + 1$, can be used to define a new, coarse-grained sublattice 
\(\Lambda'=\langle  {y'}^{\pm1}, {z'}^{\pm1}\rangle\) with generators \(y'=y^{\alpha_1-\alpha^*_1}z^{\alpha_2-\alpha^*_2}\) and
\(z'=y^{m(\alpha_2-\alpha^*_2)}z^{m(\alpha_1^*-\alpha_1)}\), and whose origin we set at $\boldsymbol{\alpha}^*$, as shown in Fig.~\ref{fig:sublattice}. It will later be useful to write $L_y'$ and $L_z'$ for number of distinct sites of $\Lambda'$ in the $y'$ and $z'$ directions respectively. Now we explain this choice of sublattice.}

\resub{We call two different $X$-type generators sharing common qubits
an ``$X$-$X$ interaction''. The result of our coarse-graining procedure is that generators associated with the sublattice $\Lambda'$ only have $X$-$X$ interactions between nearest neighbors in the $y'$ direction and along horizontal edges. We will show this in several stages.
We first establish that $X$-$X$ interactions arise between nearest neighbors and horizontal edges in the $y'$ direction. To see this, consider the $X$-type generator $A_{\boldsymbol{\alpha}^*}$ and its nearest neighbor $A_{\boldsymbol{\alpha}}$ without loss of generality, and note that $f$ and $g$ characterize their support on horizontal and vertical edges respectively.
For horizontal edges, we have $\operatorname{supp}(y^{\alpha_1}z^{\alpha_2}f)\cap \operatorname{supp}(y^{\alpha^*_1}z^{\alpha^*_2}f) = \{\boldsymbol{\alpha}^* + \boldsymbol{\alpha}\}$ by definition of $\boldsymbol{\alpha}^*$ and $\boldsymbol{\alpha}$. For the vertical edges, assume that \(A_{\boldsymbol{\alpha}^*}\) and  \(A_{\boldsymbol{\alpha}}\) share common qubits on a vertical edge. 
Then there exists $\boldsymbol{\tilde \beta} \in \operatorname{supp}(y^{\alpha_1}z^{\alpha_2}g)\cap \operatorname{supp}(y^{\alpha^*_1}z^{\alpha^*_2}g) \neq \emptyset$ such that $\boldsymbol{\tilde \beta}-\boldsymbol{\alpha}^*, \boldsymbol{\tilde \beta}-\boldsymbol{\alpha} \in \operatorname{supp}(g)$, which contradicts our assumptions. Therefore \(A_{\boldsymbol{\alpha}}\) shares exactly one qubit with \(A_{\boldsymbol{\alpha}^*}\). Next we rule out $X$-$X$ interactions between nearest neighbors in the $z'$ direction. Thus consider the nearest neighbor $z' A_{\boldsymbol{\alpha}^*} = A_{\boldsymbol{\alpha}^*+\mathbf{z}'}$. For vertical edges, we observe that
$\| \mathbf{z}' \|_2 \geq m\left|\alpha_2-\alpha_2^*\right|> 2\|\boldsymbol{\beta}\|_2$, which implies no overlapping qubits by definition of $\boldsymbol{\beta}$. For horizontal edges, the assumption that $g$ has non-empty support sets $m \geq 2$ and therefore $\|\mathbf{z}'\|_2 = m\|\boldsymbol{\alpha}-\boldsymbol{\alpha}^*\|_2 > d(f)$, which again rules out overlapping qubits. Finally, the triangle inequality rules out $X$-$X$ interactions beyond nearest neighbors.}



\resub{Having established the desired properties of $\Lambda'$, we now deduce their consequences for the parity function $c(\mathbf{x})$ of interest. For a given stabilizer $M$, write \(x_s(\alpha,\beta)\) for the Boolean variable corresponding to the presence or absence of $A_{(\alpha,\beta)}$ in $M$ and \(x_p(\alpha,\beta)\) for the
Boolean variable corresponding to the presence or absence of \(B_{(\alpha,\beta)}\) in $M$, as in Eq.~\eqref{eq:stab_op} above. According to Eq.~\eqref{eq:simplified-parity-function}, the polynomial factor of the cubic term arising from $x_s({\boldsymbol{\alpha}^*})x_s({\boldsymbol{\alpha}})$ in $c(\mathbf{x})$ will consist solely of contributions from $Z$-type generators that overlap with both $A_{\boldsymbol{\alpha}^*}$ and $A_{\boldsymbol{\alpha}}$ at the edge $y^{\alpha_1+\alpha^*_1}z^{\alpha_2+\alpha^*_2}e_1$. These can be read off from the excitation map of a Pauli $X$ operator at $y^{\alpha_1+\alpha^*_1}z^{\alpha_2+\alpha^*_2}e_1$, namely $B_{(0,0)} \cdot (y^{\alpha_1+\alpha^*_1}z^{\alpha_2+\alpha^*_2},0\mid 0,0)^{\mathsf T} = y^{\alpha_1+\alpha^*_1}z^{\alpha_2+\alpha^*_2}g$, which tells us the translates of $B_{(0,0)}$ that overlap with the edge in question. In particular, the polynomial factor of $x_s({\boldsymbol{\alpha}^*})x_s({\boldsymbol{\alpha}})$ in $c(\mathbf{x})$ is given by $\bigoplus_{\boldsymbol{\alpha} \in \operatorname{supp}(y^{\alpha_1+\alpha^*_1}z^{\alpha_1+\alpha^*_1}g)} x_p(\boldsymbol{\alpha})$. Next consider translating along direction $\mathbf{y}'$ of the sublattice $\Lambda'$ by $j$ sites. Let $\xi_j = \operatorname{supp}(y^{j(\alpha_1-\alpha_1^*)+\alpha_1^*}z^{j(\alpha_2-\alpha_2^*)+\alpha_2^*}g(y,z))$ correspond to the sites of $Z$-type generators in the polynomial factor of $x_s({\boldsymbol{\alpha}^*}+j\mathbf{y}')x_s({\boldsymbol{\alpha}}+j\mathbf{y}')$. It follows that distinct \(\xi_j\) have no common horizontal edges for the same reason that there are no $X$-$X$ interactions on vertical edges.}


\resub{Finally, we apply this  sublattice structure to estimating the bound in Eq.~\eqref{eq:derivative-bound2}.  We use the subscripts \(s\) and \(p\) to indicate generator type, 
and write \(\Lambda_s\) and \(\Lambda_p\) for the 
sets of Boolean variables $\{x_s(\alpha,\beta)\}_{(\alpha,\beta)\in\Lambda}$ and
$\{x_p(\alpha,\beta)\}_{(\alpha,\beta)\in\Lambda}$, respectively. 
First, consider the principal submatrix induced by a chain of $X$-type generators $\langle {y'}^{\pm1}\rangle_s$, which we truncate by one element to exclude interactions crossing the boundary. Let $\mathbf{a} \in \mathbb{F}_2^{n}$ denote the Boolean direction in which the derivative $D_{\mathbf{a}}c$ is being taken. Consider the sum over $Z$-type generators in $\xi_i$, namely  $a_p({\xi_i} )\coloneqq \oplus_{\boldsymbol{\gamma}\in {\xi_{i}}} a_{\boldsymbol{\gamma},p}$. Then each nonzero $a_p({\xi_i} )$ increases the rank of the submatrix by at least $1$ and the probability that $a_p({\xi_i} )=1$ is $\frac{1}{2}$ under uniformly random assignment of $\mathbf{a}$ due to pairwise disjointness of the $\xi_i$.
Since these chains are non-interacting along $\mathbf{z}'$, the above arguments extend to the principal submatrix induced by the whole truncated sublattice $\Lambda''_s$. Denote $\xi_i$ shifted by ${z'}^j$ by an extra superscript $\xi_i^j$, and let $\Xi$ be the union of $\Xi^j = \bigcup_{1 \leq i \leq L_y'-1} \xi_i^j$ over the truncated sublattice $\Lambda''$ of $L_y'-1$ by $L_z'$ sites. Thus $\Xi$ contains $N_{\Xi} = \eta L_y \times L_z$ sites, where 
$\eta=\frac{|\operatorname{supp}(g)|\,|\Lambda''|}{L_yL_z}=\frac{|\operatorname{supp}(g)|}{m\,d(f)^2}
+O\!\left(\frac1{L_y},\frac1{L_z}\right)$. Let $\mathbf{a}_p(\Xi)\coloneqq \mathbf{a}|_{\Xi_p}$ and $\mathbf{a}_p(\xi)$ denote the list of $a_p({\xi_i^j} )$. Then
\begin{equation}
\operatorname{rank}(D_{\mathbf a}c)
\;\ge\;
\operatorname{rank}\!\left(\left.(D_{\mathbf a}c)\right|_{\Lambda''_s\times\Lambda''_s}\right)
\;\ge\;
|\mathbf{a}_p(\xi)|,
\end{equation}
    where $|\mathbf{a}_p(\xi)|$ denotes the Hamming weight of $\mathbf{a}_p(\xi)$. Applying Eq. (57b), we get

    \begin{align}
p_{\mathrm{cl}}^*
&\le \tfrac{1}{2}
+ \tfrac{1}{2}
\sqrt{\frac{1}{2^{N_{\Xi}}}
\sum_{\mathbf{a}_p(\Xi) \in \mathbb{F}_2^{N_{\Xi}}}
2^{-\frac{|\mathbf{a}_p(\xi)|}{2}}} \\
&= \tfrac{1}{2}
+ \tfrac{1}{2}
\sqrt{\frac{1}{2^{|\Lambda''|}}
\sum_{k=0}^{|\Lambda''|} \binom{{|\Lambda''|}}{k} 2^{-\frac{k}{2}}} \\
&= \tfrac{1}{2}
+ \tfrac{1}{2}
\left( \cos{\frac{\pi}{8}} \right)^{|\Lambda''|},
\end{align}
 where $|\Lambda''|\geq
\frac{(L_y-4D)(L_z-4D)}{m\,d(f)^2}\geq
\frac{L_yL_z}{4m\,d(f)^2}$.
The first inequality follows by averaging over the sublattice and using
$|\alpha_\mu-\alpha_\mu^*|\leq 2D$ for $\mu=1,2$, and the second uses $L_y,L_z\geq 8D$.} 
\end{proof}



\begin{figure}[t]
        \centering
\includegraphics[width=0.5\linewidth]{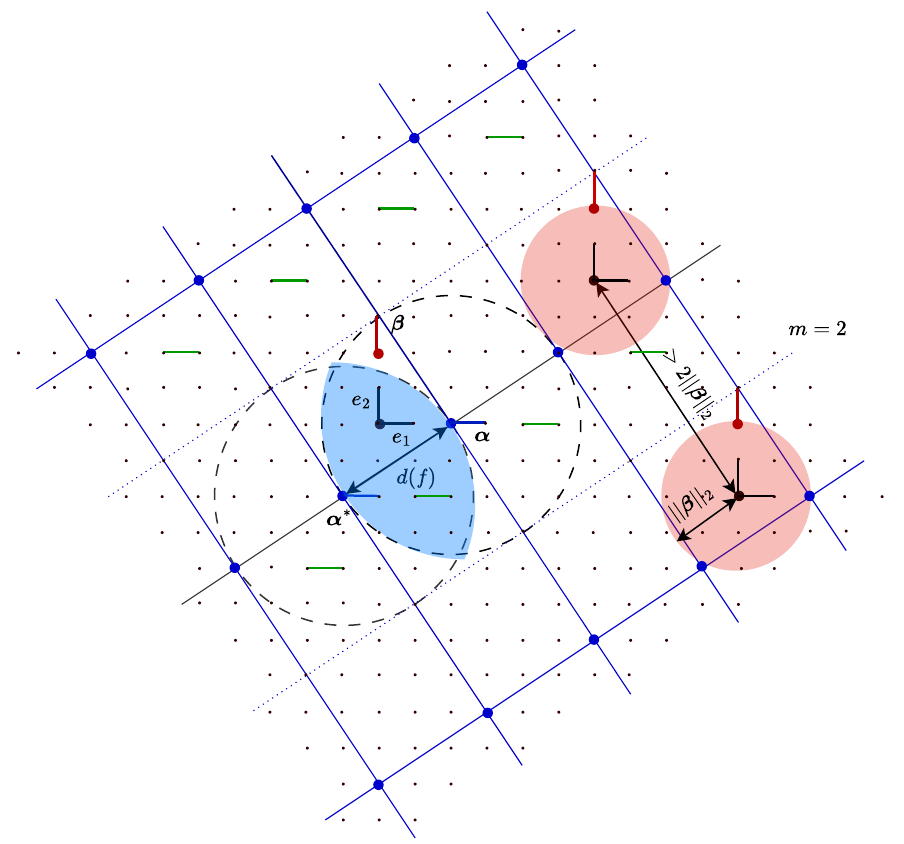}
    \caption{\resub{The rectangular sublattice induced by the generators in Eq.~\ref{eq:ldpc}, under the assumptions listed in Theorem~\ref{thm:2d translation invariant}. Each unit cell of the original square lattice (black grid) contains two qubits living on the two edges $e_1$ and $e_2$ of the unit cell at the origin. The original lattice is ``coarse-grained'' so that the supports of the $X$-type generators on the sublattice (blue grid) overlap only on the horizontal edges (green lines) shared by nearest neighbors along one direction. Two monomials in $f$ whose exponents $\boldsymbol{\alpha}$ and $\boldsymbol{\alpha}^*$ are farthest apart define one of the basis vectors of the sublattice, $\boldsymbol{\alpha}-\boldsymbol{\alpha}^*$. The lens-shaped blue region illustrates the range of sites over which terms of $f$ can lie. The monomial in $g$ whose exponent $\boldsymbol{\beta}$ is farthest from the lattice origin sets the unit-cell length of the sublattice along the direction perpendicular to $\boldsymbol{\alpha}-\boldsymbol{\alpha}^*$, chosen to be greater than $2\|\boldsymbol{\beta}\|_2$. The red circular regions illustrate the range of sites over which terms of $g$ can lie after two successive translations.}}
      \label{fig:sublattice}
    \end{figure}

\section{Transfer-matrix bounds for cyclic cluster states}
\label{sec:cluster}
The stabilizer-testing game for cyclic cluster states
\begin{equation}
Z_{i-1} X_i Z_{i+1} |C_n\rangle = |C_n\rangle, \quad i=1,2,\ldots,n,
\end{equation}
on a ring $n+1 \equiv 1$ was first introduced as a Bell inequality in Ref. \cite{guhne2005bell} before its formulation as a nonlocal game in Ref. \cite{Cabello_2013}, and subsequent six-qubit realization on a trapped-ion quantum computer~\cite{Daniel_2022}. Classical values of this stabilizer-testing game have been computed for $n \leq 10$ and it is further known~\cite{guhne2005bell} that $p_{\mathrm{l.h.v.}}^*(G)<1$ for all $n$, but the classical value of this game has not been characterized further for $n>10$~\cite{Daniel_2022}. In this section, we introduce an efficiently computable transfer-matrix expression for the winning probability attained by any given classical strategy for this game and use this to derive asymptotically tight lower and upper bounds on $p_{\mathrm{cl}}^*(G)$ for all $n \geq 3$. We also use this to obtain exact classical values of this game for $3 \leq n \leq 16$ qubits or players via exhaustive search, see Table \ref{tab:cycliccluster} in Appendix \ref{app:numerics}. We note that cyclic cluster states are amenable to a transfer-matrix approach by virtue of their one-dimensional and translationally invariant nature.

Thus let $S$ denote the full stabilizer group of $|C_n\rangle$, $g_j = Z_{j-1}X_jZ_{j+1}$ its generators and let $M = \prod_{j=1}^n g_j^{x_j} \in S$ be a query in the stabilizer-testing game, which we specify by the bit-string $\mathbf{x} \in \{0,1\}^n$ as in Section~\ref{sec:codingbounds}. The corresponding element $c(\mathbf{x})$ of the parity vector $\mathbf{c}$ can be expressed as a cyclic (also known as rotation-symmetric~\cite{cusick2000fastevaluationweightsnonlinearity}) and cubic Boolean function of $\mathbf{x}$, namely~\cite{Daniel_2022}
\begin{equation}
\label{eq:cyclicB}
c({\mathbf{x}}) \equiv R_3(\mathbf{x}) := \sum_{j=1}^n x_{j-1}x_j x_{j+1} \pmod{2}.
\end{equation}
Organizing the columns of the incidence matrix as
\begin{equation}
A^{\mathbbm{1}}_j(\mathbf{x}) = A_{\mathbf{x},4(j-1)+1}, \quad  A^{X}_j(\mathbf{x}) = A_{\mathbf{x},4(j-1)+2}, \quad A^{Y}_j(\mathbf{x}) = A_{\mathbf{x},4(j-1)+3}, \quad A^{Z}_j(\mathbf{x}) = A_{\mathbf{x},4(j-1)+4},
\end{equation}
we can similarly express these as the quadratic Boolean functions
\begin{align}
A^{\mathbbm{1}}_j(\mathbf{x}) &\equiv (1+x_j)(1+x_{j-1}+x_{j+1}) \pmod{2}, \\
A^{X}_j(\mathbf{x}) &\equiv x_j(1+x_{j-1}+x_{j+1}) \pmod{2},\\
A^{Y}_j(\mathbf{x}) &\equiv x_j(x_{j-1}+x_{j+1}) \pmod{2}, \\
A^{Z}_j(\mathbf{x}) &\equiv (1+x_j)(x_{j-1}+x_{j+1}) \pmod{2}.
\end{align}
Let us also write $b^{\mathbbm{1}}_j= b_{4(j-1)+1}$,  $b^X_j = b_{4(j-1)+2}$, $b^Y_j = b_{4(j-1)+3}$, $b^Z_j = b_{4(j-1)+4}$. Using the above expressions and collecting prefactors, we find that
\begin{align}
\label{eq:formulaforAb}
A \mathbf{b} \equiv \resub{w+}\sum_{j=1}^n u_jx_j + v_j x_j(x_{j-1}+x_{j+1}) \pmod{2},
\end{align}
where $u_j := b^{\mathbbm{1}}_{j-1} + b^{\mathbbm{1}}_{j}+b^{\mathbbm{1}}_{j+1} + b^Z_{j-1} + b^X_j + b^Z_{j+1}$, $v_j := b_j^{\mathbbm{1}} + b_j^X+b_j^Y + b_j^Z$ \resub{and the constant $w :=\sum_{j=1}^{n} b_j^{\mathbbm{1}}$}.

Now for a given classical strategy $\mathbf{b}$, let $f_n(\mathbf{b})$ denote the difference between the number of satisfied and unsatisfied rows of Eq. \eqref{eq:classconstr}, so that $f_n(\mathbf{b}) = 2^{n} - 2d_{\mathrm{H}}(A\mathbf{b},\mathbf{c})$ and $p_{\mathrm{cl}}(\mathbf{b}) = \frac{1}{2}+\frac{1}{2} \frac{f_n(\mathbf{b})}{2^n}$. Then Eq. \eqref{eq:formulaforAb} implies that
\begin{equation}
\label{eq:fbfull}
f_n(\mathbf{b}) = \resub{(-1)^w} \sum_{\mathbf{x}\in \{0,1\}^n} \prod_{j=1}^n (-1)^{x_{j-1}x_jx_{j+1}+u_j x_j + v_j x_j(x_{j-1}+x_{j+1})}.
\end{equation}
This expression can be interpreted physically as the partition function for an interacting, disordered one-dimensional classical Ising model evaluated at an imaginary temperature $\beta = i \pi$. This physical analogy suggests a transfer-matrix formulation, and indeed we can write
\begin{equation}
\label{eq:fbtm}
f_n(\mathbf{b}) = \resub{(-1)^w} \mathrm{tr}\left[\prod_{j=1}^n T(u_j,v_j)\right], \quad T(u,v) =  \begin{pmatrix}
    1 & 1 & 0 & 0 \\ 
    0 & 0 & (-1)^{u} & (-1)^{u+v} \\
    1 & 1 & 0 & 0 \\  0 & 0 & (-1)^{u+v} & (-1)^{u+1}
\end{pmatrix},
\end{equation}
where the definition of the three-site (or ``two-bond'') transfer matrices 
\begin{equation}
T_{x_{j-1}x_j,x_{j}'x_{j+1}}(u_j,v_j) = \delta_{x_j x'_j} (-1)^{x_{j-1}x_jx_{j+1}+u_jx_j +v_j(x_{j-1}+x_{j+1})}
\end{equation}
follows directly from Eq. \eqref{eq:fbfull}. This expression can be simplified using the fact that the matrices $T(u,v)$ always have rank three, via the QR decomposition 
\begin{equation}
\begin{pmatrix} 1 & 1 & 0 & 0 \\
 0 & 0 & (-1)^u & (-1)^{u+v} \\
 1 & 1 & 0 & 0 \\
 0 & 0 & (-1)^{u+v} & (-1)^{u+1}\end{pmatrix} = \frac{1}{\sqrt{2}}\begin{pmatrix} 1 & 0 & 0 \\ 0 & (-1)^u & (-1)^{u+v} \\ 1 & 0 & 0 \\ 0 & (-1)^{u+v} & (-1)^{u+1} \end{pmatrix}\begin{pmatrix} \sqrt{2} & \sqrt{2} & 0 & 0 \\
 0 & 0 & \sqrt{2} & 0 \\
 0 & 0 & 0 & \sqrt{2}
 \end{pmatrix}
\end{equation}
that holds for any $u,v \in \{0,1\}$. Substituting this QR decomposition into Eq. \eqref{eq:fbtm} and applying cyclicity of the trace then yields an expression
\begin{align}
\label{eq:fbrtm}
f_n(\mathbf{b}) =\resub{(-1)^w}\mathrm{tr}\left[\prod_{j=1}^n \tilde{T}(u_j,v_j)
\right], \quad \tilde{T}(u,v) = \begin{pmatrix}
    1 & (-1)^{u} & (-1)^{u+v} \\ 1 & 0 & 0 \\ 0 & (-1)^{u+v} & (-1)^{u+1}
\end{pmatrix}
\end{align}
in terms of ``reduced'' transfer matrices $\tilde{T}(u,v)$. \resub{We deduce that for this stabilizer-testing game,
\begin{equation}
\label{eq:pclstarcluster}
 p^*_{\mathrm{cl}}(G) \leq \frac{1}{2} + \frac{1}{2^{n+1}} \max_{\mathbf{u},\mathbf{v}\in \{0,1\}^{n}}\left|\mathrm{tr}\left[\prod_{j=1}^n \tilde{T}(u_j,v_j)
\right]\right|,
\end{equation}}
so that a na{\"i}ve optimization over $2^{4n}$ variables reduces to an optimization over $2^{2n}$ variables. Nevertheless, the primary advantage conferred by this transfer-matrix representation is that it allows for an analytical derivation of asymptotically tight upper and lower bounds on $p_{\mathrm{cl}}^*(G)$ for all $n$. First we derive a lower bound that corresponds to the Hamming weight of the cyclic Boolean function Eq. \eqref{eq:cyclicB}, a quantity that was considered in the context of stabilizer-testing games in Ref.~\cite{Cabello_2013} and computed explicitly in Ref.~\cite{cusick2000fastevaluationweightsnonlinearity}. Our transfer-matrix formula yields an elementary derivation of this Hamming weight, as follows.
\begin{prop}
\label{prop:cslowerbound}
Let $\lambda_0 \approx 1.769,\, \lambda_{\pm} \approx -0.885 \pm 0.590i$ denote the three roots of the cubic equation $
\lambda^3 - 2 \lambda - 2 = 0$. Then
\begin{equation}
\label{eq:pcllowerbound}
p_{\mathrm{cl}}^*(G) \geq \frac{1}{2} + \frac{1}{2^{n+1}}\left(\lambda_0^n + \lambda_+^n + \lambda_-^n\right).
\end{equation}
\end{prop}
\begin{proof}
Consider the classical strategy $\mathbf{b}=0$ that returns zero for all questions. By Eq. \eqref{eq:fbrtm}, this strategy satisfies
\begin{equation}
f_n(\mathbf{0}) = \mathrm{tr}\left[\tilde{T}(0,0)^n\right] = \lambda_0^N + \lambda_+^N + \lambda_-^N,
\end{equation}
where $\{\lambda_0,\lambda_\pm\}$ are the eigenvalues of $\tilde{T}(0,0)$ and therefore correspond to the three distinct roots of its characteristic equation $\lambda^3 - 2\lambda - 2 = 0$.
\end{proof}
However, in order to obtain an experimentally useful Bell inequality for $|C_n\rangle$, we require an upper bound on $p_{\mathrm{cl}}^*(G)$. This is complicated by the fact that $\tilde{T}(u,v)$ is not a normal matrix, so that its singular values do not coincide with its eigenvalues. We first note that a crude (but still asymptotically tight) upper bound comes from the singular-value spectrum of $\tilde{T}$, as follows.
\begin{prop}
Let $s_1 = \sqrt{2+\sqrt{2}}$, $s_2 = \sqrt{2}$ and $s_3 = \sqrt{2-\sqrt{2}}$. Then
\begin{equation}
\label{eq:svupperbound}
p_{\mathrm{cl}}^*(G) \leq \frac{1}{2} + \frac{1}{2^{n+1}} (s_1^n + s_2^n + s_3^n).  
\end{equation}
\end{prop}
\begin{proof}
Note that all the matrices $\{\tilde{T}(u_j,v_j)\}_{j=1}^N$ in Eq. \eqref{eq:fbrtm} have the same singular values, since
\begin{equation}
\tilde{T}(u,v)\tilde{T}(u,v)^\dagger = \begin{pmatrix} 3 & 1 & 0 \\ 1 & 1 & 0 \\ 0 & 0 & 2
\end{pmatrix}
\end{equation}
is independent of $u$ and $v$. This matrix has characteristic equation $(\lambda-2)(\lambda-2+\sqrt{2})(\lambda-2-\sqrt{2})=0$, which implies that the singular values of $\tilde{T}(u,v)$ are $(s_1,s_2,s_3) = (\sqrt{2+\sqrt{2}},\sqrt{2},\sqrt{2-\sqrt{2}})$ for all $u,v\in \{0,1\}$. Now fix an arbitrary classical strategy $\mathbf{b}$,  let $(\Lambda_1,\Lambda_2,\Lambda_3)$ denote the eigenvalues of $\prod_{j=1}^n \tilde{T}(u_j,v_j)$ ordered such that $|\Lambda_1| \geq |\Lambda_2| \geq |\Lambda_3|$ and let $\Sigma_1 \geq \Sigma_2 \geq \Sigma_3$ denote its singular values. Then $|\Lambda_j| \leq \Sigma_j \leq s_j^n$ by standard results in matrix analysis~\cite{Marshall2011}, which yields the uniform upper bound 
\begin{equation}
|f_n(\mathbf{b})| = |\Lambda_1+\Lambda_2+\Lambda_3| \leq s_1^n + s_2^n + s_3^n.
\end{equation}
\end{proof}
Thus our rough singular-value estimate already suffices to capture exponentially fast convergence of $p_{\mathrm{cl}}^*(G)$ to its asymptotic value of $1/2$ as $n \to \infty$. However, the leading asymptotic behaviour $\sim \frac{1}{2}+ \frac{1}{2}(0.923...)^n$ of the upper bound Eq. \eqref{eq:svupperbound} decays slower than the lower bound $\sim \frac{1}{2}+ \frac{1}{2}(0.884...)^n$ which is found empirically to be tight for $4 < n \leq 16$ (see Table \ref{tab:cycliccluster}), suggesting that our singular-value estimate may be looser than the slowest possible decay of the leading eigenvalue $\Lambda_1$ defined above. This indeed turns out to be the case.
\begin{theorem}\label{th2}
Let $\lambda_0 \approx 1.769$ denote the real root of $\lambda^3-2\lambda-2=0$ as in Proposition \ref{prop:cslowerbound}. Then
\begin{equation}
\label{eq:eigenvalueupperbound}
p_{\mathrm{cl}}^*(G) \leq \frac{1}{2} + \frac{1}{2^{n+1}}(3\lambda_0^n).
\end{equation}
\end{theorem}

The proof of Theorem \ref{th2} relies on the following definitions and results:

\begin{defn}\emph{(Joint spectral radius)}
Let $\rho(\cdot)$ denote the usual spectral radius of a matrix. Two closely related definitions of the joint spectral radius of a finite set of matrices $\Sigma$ are
\begin{equation}
\rho_t(\Sigma) :=
 \max\limits_{(T_1,...,T_t) \in \Sigma^t} \left\{ \left(\rho\left( \prod_{i=1}^t T_i\right)\right)^{1/t}\right\}
 \end{equation}
and
\begin{equation}
\hat{\rho}(\Sigma) :=
 \lim\limits_{t\rightarrow \infty} \rho_t(\Sigma).
 \end{equation}
\end{defn}
We have
\begin{prop}\label{prop3}\emph{\citep[Proposition 1.6]{jungers2009joint}}
For any $t>0$, $\rho_t(\Sigma) \leq \hat{\rho}(\Sigma)$.
\end{prop}

We need the following to calculate the joint spectral radius $\hat{\rho}(\Sigma)$:
\begin{defn}\emph{(Invariant polytope)}
An invariant polytope for a set of matrices $\Sigma$ is a convex polytope $P$ such that $r P = \mathrm{Conv}(\Sigma P)$ for some $r>0$.
\end{defn}
and we have 
\begin{prop}\label{prop4}\emph{\citep[Theorem 1]{protasov2005geometric}}
Suppose $\Sigma$ is a finite matrix set whose invariant polytope $P$ satisfies $r P = \mathrm{Conv}(\Sigma P)$. Then $\hat{\rho}(\Sigma) = r$.
\end{prop}

Below we apply Proposition \ref{prop4} to the problem at hand.
\begin{prop}\label{prop5}
Define $\Sigma := \{\tilde{T}(u,v)\}_{u,v = \pm 1}$. We have $\hat{\rho}(\Sigma) = \lambda_0$.
\end{prop}

\begin{proof}

\begin{figure}[t]
\centering
\includegraphics[width=0.45\textwidth]{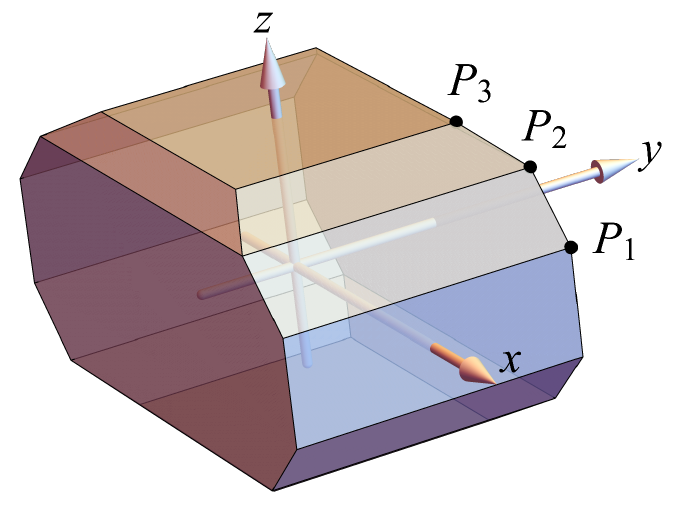}
\caption{Invariant polytope $P$ defined by Eq.~\eqref{p1p2p3}.}\label{invariant_polytope}
\end{figure}

We prove this by explicitly constructing the invariant polytope $P$ of $\Sigma$. Define $s(\lambda_0) := \sqrt{4 \lambda_0 ^2+8 \lambda_0 +6}$, and points $P_{1,2,3} \in \mathbb{R}^3$:
\begin{equation}\label{p1p2p3}
\begin{aligned}
P_1 &= \frac{1}{s( \lambda_0)}\left(  \lambda_0  ( \lambda_0 +1),  \lambda_0 +1, 1\right),\\
P_2 &= \frac{1}{s( \lambda_0)} \left( \lambda_0 +2,  \lambda_0 +1, \frac{ \lambda_0 +2}{ \lambda_0 }\right),\\
P_3 &= \frac{1}{s( \lambda_0)} \left( \frac{2 \left( \lambda_0 ^2+ \lambda_0 -1\right)}{ \lambda_0 ^2}, \frac{ \lambda_0 +2}{ \lambda_0 },  \lambda_0 +1 \right).
\end{aligned}
\end{equation}
By direct calculation, one can verify that a symmetric convex polytope $P$ spanned by $P_i \in \{P_1,P_2,P_3\}$ and their mirror reflections about $xy$, $xz$, and $yz$ planes (so 24 vertices in total; see Fig.~\ref{invariant_polytope}) is the invariant polytope for $\Sigma$, with $\lambda_0 P = \mathrm{conv}(\Sigma P)$. By Proposition \ref{prop3}, $\hat{\rho}(\Sigma) =  \lambda_0$. 
\end{proof}

With these we can finally prove Theorem \ref{th2}.
\begin{proof}[Proof of Theorem~\ref{th2}]
\resub{For any deterministic classical strategy $\mathbf{b}$}, we have $p_{\mathrm{cl}}(G) = \frac{1}{2} + \frac{1}{2^{n+1}}f(\mathbf{b})
\leq
\frac{1}{2} +\frac{1}{2^{n+1}} \left|\mathrm{tr} \left[ \prod_{j=1}^n \tilde{T}(u_j,v_j)\right]\right| \leq 
\frac{1}{2} +\frac{3}{2^{n+1}}  (\rho_n(\Sigma))^n
\leq
\frac{1}{2} +\frac{3}{2^{n+1}}  \lambda_0^n
$, where the first inequality follows from the definition of the joint spectral radius, and the second from Propositions \ref{prop3} and \ref{prop5}.
\end{proof}
 
\begin{figure}[t]
    \centering
\includegraphics[width=0.7\linewidth]{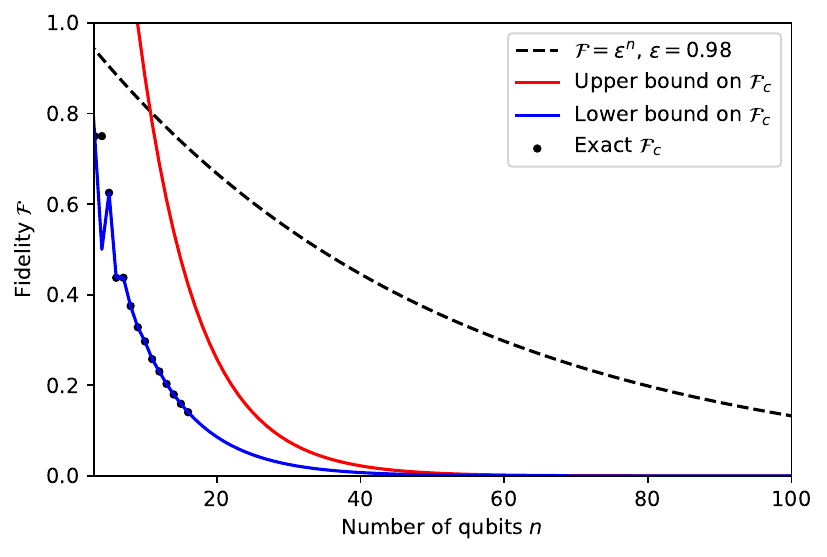}
    \caption{Comparison between our rigorous bound, Eq. \eqref{eq:clusterstatefidelity}, on the fidelity $\mathcal{F} > \mathcal{F}_c = 2p_{\mathrm{cl}}^*-1$ to the cyclic cluster state required to witness contextuality (red line) and a plausible~\cite{Iqbal_2024,kumar2025quantumclassicalseparationboundedresourcetasks} estimate of the global fidelity $\mathcal{F} = \epsilon^n,\, \epsilon = 0.98$ to the $n$-qubit cluster state attainable on state-of-the-art quantum devices (black dotted line). The broad region above the red curve and below the black dotted line suggests that witnessing contextuality of the cluster state in many-body regimes of $n=50-100$ qubits should be feasible on current quantum devices. For completeness, we include the rigorous lower bound Eq. \eqref{eq:pcllowerbound} (blue line) and exact numerical results for $n \leq 16$ (black dots and Table \ref{tab:cycliccluster}) for $\mathcal{F}_c = 2p_{\mathrm{cl}}^*-1$ that suggest tightness of this lower bound.}
    \label{fig:experiment}
\end{figure}

Our results demonstrate that the classical value of the stabilizer-testing game for the cyclic cluster state satisfies $p_{\mathrm{cl}}^*(G) - \frac{1}{2} = \mathcal{O}(0.885^n)$ as $n \to \infty$, yielding behaviour that is qualitatively similar to but quantitatively distinct from the usual Mermin-Brassard-Broadbent-Tapp parity game for the GHZ state~\cite{mermin1990extreme,brassard2004recasting}. Based on our numerical results in Table \ref{tab:cycliccluster}, we conjecture that the lower bound in Eq. \eqref{eq:pcllowerbound} is in fact tight for all $n>4$.

Finally, we note that the upper bound Eq. \eqref{eq:eigenvalueupperbound} implies by Eq. \eqref{eq:fidelityformula} that an exponentially small fidelity
\begin{equation}
\label{eq:clusterstatefidelity}
\mathcal{F} > 3(\lambda_0/2)^n = 3(0.884...)^n
\end{equation}
to the cyclic cluster state will suffice to witness its contextuality. This is remarkable given that global fidelities attainable for state preparation on state-of-the-art quantum devices are reported~\cite{Iqbal_2024,kumar2025quantumclassicalseparationboundedresourcetasks} to be as high as $\mathcal{F} \approx \epsilon^n$ with $\epsilon \approx 0.98-0.99$ for tens of qubits. The implication of Eq. \eqref{eq:clusterstatefidelity} is that a fidelity-per-qubit $\epsilon \gtrsim 0.89$ will suffice to witness contextuality of the cyclic cluster state for sufficiently large $n$, though larger values of $\epsilon$ would be required to achieve this for currently accessible numbers of qubits. Making an estimate $\epsilon = 0.98$ slightly smaller than the fidelity-per-qubit $\epsilon \gtrsim 0.984$ reported for preparing a highly non-trivial 27-qubit wavefunction in Ref.~\cite{Iqbal_2024} (including readout error) and consistent with the fidelity reported for a 71-qubit GHZ state in Fig. S6 of Ref.~\cite{kumar2025quantumclassicalseparationboundedresourcetasks}, we predict that witnessing contextuality of the cyclic cluster state should be feasible on state-of-the-art quantum devices of $n=50-100$ physical qubits, see Fig. \ref{fig:experiment}. 

We emphasize that the statistical error in estimating the fidelity by playing the game repeatedly with freshly prepared cluster states and using Eq. \eqref{eq:fidelityformula} is determined by the number of times the game is played, and is therefore independent of $n$ whenever the experimentally measured fidelity is clearly distinguishable from the exponentially small threshold in Eq. \eqref{eq:clusterstatefidelity}, $\mathcal{F} \gg 3(0.884...)^n$. In this intermediate regime of system sizes, we expect contextuality of the cyclic cluster state to be robust enough to be witnessed experimentally. Fig. \ref{fig:experiment} indicates that this regime is within reach.

Note that from an experimental viewpoint, the stabilizer-testing game for the cluster state is more robust than that for the GHZ state, for which a fidelity $\mathcal{F} > \frac{1}{2}$ is required to witness contextuality, in contrast to the exponentially small number in Eq. \eqref{eq:clusterstatefidelity}. Similarly, the quantum probability of winning the standard parity game~\cite{mermin1990extreme,brassard2004recasting} by applying its default protocol to an arbitrary state is penalized by the fidelity of that state to the odd-parity GHZ state~\cite{bulchandani2023playing}, because this game queries only a subset of the full stabilizer group of the GHZ state, see Eq. \eqref{eq:fidelityformulaGHZ}. No such penalty appears in Eq. \eqref{eq:fidelityformula}.

\section{Connection to measurement-based quantum computing}
\label{sec:MBQC}
We now elucidate the connection between our results and measurement-based quantum computation (MBQC)~\cite{raussendorf2001one}. Recall that a measurement-based quantum computation consists of a quantum resource state (typically a cluster state) to which a sequence of (typically adaptive) measurements are applied. However, deciding which measurement bases to use at any given time step based on previous measurement outcomes, as well as converting measurement outcomes into computational output, requires access to a classical control computer that performs classical side-processing. In the standard formulation of MBQC known as ``l2-MBQC''~\cite{Anders_2009,PhysRevA.88.022322}, all such classical side-processing is constrained to be linear modulo two. It was shown previously~\cite{Hoban_2011} that l2-MBQC is capable of computing any Boolean function $f : \{0,1\}^n \to \{0,1\}$ on $n$ input bits (albeit with a worst-case overhead of $2^{n}-1$ qubits). It is also known that an l2-MBQC that evaluates any non-linear Boolean function $f$ with success probability~\cite{PhysRevA.88.022322}
\begin{equation}
\label{eq:Rausseq}
p_{\mathrm{succ}} > 1-\frac{\nu(f)}{2^n}    
\end{equation}
is contextual, where $\nu(f)$  denotes the Hamming distance from $f$ to the nearest linear Boolean function. For the specific example of bent functions and even $n$, this contextuality threshold is an exponentially small improvement over random guessing, namely~\cite{PhysRevA.88.022322}
\begin{equation}
\label{eq:expsmallbent}
p_{\mathrm{succ}} > \frac{1}{2}+\frac{1}{2^{n/2+1}} = \frac{1}{2} + \frac{1}{2}(0.707...)^n.
\end{equation}

Our results connect to this understanding of MBQC in two different ways. First, our analysis of the cyclic cluster state in Section \ref{sec:cluster} suggests that computing the cyclic, cubic Boolean function $R_3(\mathbf{x})$ specified in Eq. \eqref{eq:cyclicB} in l2-MBQC yields qualitatively similar behaviour to Eq. \eqref{eq:expsmallbent}. To see this, let $R_3(\mathbf{x})$ be the cyclic cubic Boolean function on $n$ bits defined in Eq. \eqref{eq:cyclicB}. It was conjectured~\cite{cusick2000fastevaluationweightsnonlinearity} and later proved~\cite{ciungu2010cryptographic,zhang2010proofconjecturerotationsymmetric} that the nonlinearity $\mathrm{nl}_1(R_3)$ of this function is equal to its weight, implying that~\cite{cusick2000fastevaluationweightsnonlinearity}
\begin{equation}
\mathrm{nl}_1(R_3) = \mathrm{wt}(R_3) = 2^{n-1} - \frac{1}{2}(\lambda_0^n + \lambda_{-}^n + \lambda_+^n),
\end{equation}
with $\{\lambda_0,\lambda_{\pm}\}$ defined in Proposition \ref{prop:cslowerbound}. Now suppose that an l2-MBQC computes $R_3(\mathbf{x})$ with probability of success $p_{\mathrm{succ}}$. Since by definition $\mathrm{nl}_1(R_3) \leq \nu(R_3)$, we deduce by Eq. \eqref{eq:Rausseq} that a probability of success 
\begin{equation}
p_{\mathrm{succ}} > \frac{1}{2} - \frac{\mathrm{nl}_1(R_3)}{2^{n}} = \frac{1}{2} + \frac{1}{2^{n+1}}(\lambda_0^n + \lambda_-^n + \lambda_+^n) \sim \frac{1}{2} + \frac{1}{2}(0.884...)^n, \quad n \to \infty,
\end{equation}
is sufficient to witness contextuality of the quantum computation~\cite{PhysRevA.88.022322}. The qualitative similarity to the case of bent functions Eq. \eqref{eq:expsmallbent} is clear, with the approach to the limit slower than in Eq. \eqref{eq:expsmallbent} as expected for a non-bent function such as $R_3(\mathbf{x})$. Thus contextuality thresholds with the qualitative form of Eq. \eqref{eq:expsmallbent} can arise even for non-bent functions and for odd values of $n$.

Second, we note that given an arbitrary $n$-qubit stabilizer state $|\psi\rangle$ on $n$ qubits, with stabilizer group $S = \langle g_1,g_2,\ldots,g_n\rangle$, the default quantum strategy $\mathcal{S}_0$ for its stabilizer-testing game can be interpreted as an MBQC~\footnote{R. Raussendorf, private communication}. Specifically, given an input bit string $\mathbf{x} \in \{0,1\}^n$, this MBQC computes the parity bit $c(\mathbf{x}) \in \{0,1\}$ corresponding to the stabilizer $M = \prod_{j=1}^n g_j^{x_j}$. However, this MBQC differs from the standard formulation of l2-MQBC in two important respects: (i) the classical side-processing that relates input bit strings to measurement settings is quadratic in general, by Eqs. \eqref{eq:incidence-polynomial} and \eqref{eq:incidence-polynomial-P}, and (ii) only an input-dependent subset of the resource qubits are measured during the course of the computation, corresponding to the fact that a generic stabilizer will include identity matrices in its Pauli representation Eq. \eqref{eq:Paulirep}. An alternative way to connect our results to l2-MBQC is to consider reformulating l2-MBQC with three measurement settings per qubit (corresponding to Pauli $X$, $Y$ and $Z$) rather than the conventional two. This presumably modifies the contextuality threshold Eq. \eqref{eq:Rausseq}. We defer a more detailed analysis of the connection between our results and MBQC to future work. 

\section{Conclusion}
We have developed a systematic theoretical framework for understanding the classical values of stabilizer-testing games, that both unifies several related recent constructions in the literature~\cite{Bravyi_2018,Daniel_2021,Bulchandani_2023,bulchandani2023playing,Hart_2025,hart2025braiding} and significantly extends previous theoretical analyses of such games~\cite{guhne2005bell,Cabello_2013,Daniel_2022}. For example, an open question resolved by our work is the large $n$ behaviour of the classical value for cyclic cluster states~\cite{guhne2005bell,Cabello21,danielsen2005self}. Our solution to this problem leads to the most physically striking result of this paper, namely that imperfect fidelities-per-qubit $\epsilon \gtrsim 0.89$ (including readout and measurement errors) should be sufficient to experimentally detect contextuality of the cyclic cluster state as $n \to \infty$. This improves upon earlier related results that were sensitive to the global parity of the state~\cite{bulchandani2023playing,Bulchandani_2023}, namely Eq. \eqref{eq:fidelityformulaGHZ}. We note that while the global-parity penalty for the conventional GHZ state in Eq. \eqref{eq:fidelityformulaGHZ} does not appear to be a serious obstacle for tests of quantum contextuality  on quantum computers, given very recent demonstrations of quantum advantage for GHZ states of up to 71 qubits~\cite{kumar2025quantumclassicalseparationboundedresourcetasks}, it would be problematic for demonstrating contextuality via the natural condensed-matter realization of the GHZ state as the ground state of a quantum Ising model, for which stray longitudinal fields or temperatures above the ground-state energy gap would rapidly eliminate any such quantum advantage~\cite{bulchandani2023playing}. Thus our results provide the first unambiguous theoretical demonstration that nonlocal games querying exponentially many stabilizers of a given state are robust enough to be realized on imperfect quantum devices with an imperfect fidelity-per-qubit $\epsilon < 1$, in truly many-body regimes of $n \gg 1$ qubits.

In this respect, our work is complementary to related theoretical~\cite{hart2025braiding} and experimental~\cite{Hart_2025} results that exploit a different potential source of ``robustness'' for stabilizer-testing games, namely topological order. By contrast, the robustness that we identify in this work is a consequence of the nonquadraticity of the parity function defined by the stabilizer algebra, and thus has no straightforward interpretation in terms of the $\mathbb{Z}_2 \times \mathbb{Z}_2$ SPT order~\cite{Gu_2009,Else_2012} of the cyclic cluster state; indeed, we expect that there should exist similarly robust examples with no ``order'' at all in the conventional sense of condensed matter physics.

This improved understanding of the origins of nonclassicality for stabilizer-testing games raises the question of whether we can extend this understanding beyond qubit stabilizer states. While the latter is the most natural setting for benchmarking the preparation of quantum codewords on current quantum devices, there are other constructions that might be natural to consider in different contexts. For example, performing an analogous benchmarking procedure for qudit or continuous-variable systems would require extending our results accordingly. We expect that our Reed-Muller-code based analysis in this paper should generalize straightforwardly to qudits with local Hilbert space dimension $d>2$. The generalization to continuous variable systems with $d \to \infty$ seems less trivial, and would be interesting to consider given the rapid recent improvements in the experimental control of continuous-variable quantum systems~\cite{Gottesman_2001,Fl_hmann_2019,Campagne_Ibarcq_2020}. 

From the perspective of achieving a starker separation between quantum and classical values, it seems natural to consider states arising from higher levels $k>2$ of the Clifford hierarchy, which have been related to the evaluation of degree-$k$ Boolean polynomials in the setting of MBQC~\cite{frembs2023hierarchies}. This might provide one route towards generalizing our observations on how cubic Boolean polynomials give rise to non-zero quantum advantage for stabilizer-testing games to higher-degree Boolean polynomials, with a higher polynomial degree presumably corresponding to an enhancement of such quantum advantage.

Most challenging would be an extension of our results towards more realistic ground states of the type that arise in condensed matter physics. In this case, identifying the optimal set of commuting operators away from zero-correlation-length fixed points, such as stabilizer states~\cite{verstraete2005renormalization}, is a non-trivial problem~\cite{Daniel_2021,bulchandani2023playing,Hart_2025}. For example, if one applies a ``fixed-point protocol'', such as the default protocol of Section~\ref{sec:stabtest}, away from a given fixed point, one should strictly optimize over all possible single-qubit rotations to reliably estimate the attainable quantum advantage~\cite{miller2012optimal,lin2023quantumtasksassistedquantum}, but even this step seems intractable in practice. One possible generalization of our approach in this paper, which might be easier than embracing the full complexity of many-body quantum ground states, would be to frustration-free Hamiltonians, for which measurements of the commuting summands in the Hamiltonian would be the natural starting point for an analysis in terms of quantum contextuality.

Finally, some recent motivation for considering quantum contextuality comes from the pivotal role played by nonlocal games in the Bravyi-Gosset-K{\"o}nig demonstration of quantum computational advantage with shallow circuits~\cite{Bravyi_2018,bravyi2020quantum,watts2019exponential}. Formally, this amounts to constructing a search problem (the ``hidden linear function problem'') that lies in the class of constant-depth quantum circuits with bounded fan-in gates $\mathrm{QNC}_0$ but not in the corresponding class of classical circuits $\mathrm{NC}_0$, thereby demonstrating an unconditional separation $\mathrm{NC}_0 \subsetneq \mathrm{QNC}_0$~\cite{watts2019exponential}. The question of whether or not scalable nonlocal games, of the type we have considered here, could be used to extend these results in a nontrivial fashion merits further investigation.\\

\noindent \textbf{Note added}: while this manuscript was in preparation, we became aware of independent related work on this topic, which will appear in the same arXiv posting~\cite{toappear}. This work also estimates classical values for the cyclic cluster and toric-code states, and where our results overlap, they agree.

\section{Acknowledgments}
We thank D.S. Borgnia, F.J. Burnell, Z. Han, C. Lin, P. Roushan and S.L. Sondhi for helpful discussions on this topic. We especially thank R. Raussendorf for numerous discussions on contextuality and for pointing out the connection between our results and measurement-based quantum computing. This work was supported in part by the NOTS cluster operated by Rice University's Center for Research Computing (CRC). S.~L.~and W.~W.~H.~are supported by the National Research Foundation (NRF), Singapore, through the NRF Felllowship NRF-NRFF15-2023-0008. W.~W.~H.~is supported through the National Quantum Office, hosted in A*STAR, under its Centre for Quantum Technologies Funding Initiative (S24Q2d0009).
\bibliography{games.bib}

\appendix 
\section{Proofs of Propositions}\label{app:proofs}
\renewcommand{\theequation}{A\arabic{equation}}
\setcounter{equation}{0}
\begin{lem}[Gaussian elimination on stabilizers]\label{lem: Gaussian elimination}
Let $V=\mathbb F_2^{2n}$ be written as a direct sum of $n$ two–dimensional sites
\[
V=\bigoplus_{i=1}^n W_i,\qquad W_i\cong\mathbb F_2^2,
\]
with the usual Pauli encoding on each site $i$:
\[
(0,0)\!\leftrightarrow\! I,\quad (1,0)\!\leftrightarrow\! X,\quad (0,1)\!\leftrightarrow\! Z,\quad (1,1)\!\leftrightarrow\! Y.
\]
For any linear subspace $L\le V$ there exists a basis $\{r_1,\dots,r_k\}$ of $L$ ($k=\dim L$) such that, on every site $i\in\{1,\dots,n\}$, at most two distinct nontrivial Pauli operators occur among these generators (all the rest being $I$).
\end{lem}

\begin{proof}
The proof is essentially Gaussian elimination. To apply this, we must fix an elimination order so that operations on later sites do not affect the results already obtained on earlier sites. For each site $i$, let $\pi_i:V\to W_i$ be the coordinate projection, and define the descending chain of residual subspaces
\[
L^{(1)}:=L,\qquad
L^{(i+1)}:=\{\,v\in L^{(i)}:\ \pi_i(v)=0\,\}\quad (1\le i\le n).
\]
Write $d_i:=\dim\bigl(\pi_i(L^{(i)})\bigr)\in\{0,1,2\}$. Since $W_i\cong\mathbb F_2^2$, $d_i\le2$, and by construction
\[
\dim L^{(i+1)}=\dim L^{(i)}-d_i\quad(1\le i\le n),\qquad 
\sum_{i=1}^n d_i=\dim L=:k.
\]

\paragraph{Round I}
At step $i$ we work inside $L^{(i)}$.

\emph{Case $d_i=0$.} Do nothing and pass to $L^{(i+1)}=L^{(i)}$.

\emph{Case $d_i=1$.} Choose $r_{i,1}\in L^{(i)}$ with $\pi_i(r_{i,1})\ne0$.
For each $u\in L^{(i)}$, replace $u\mapsto u+\alpha(u)r_{i,1}$ with the unique $\alpha(u)\in\{0,1\}$ making $\pi_i(u+\alpha r_{i,1})=0$.
Since $r_{i,1},u\in L^{(i)}$, they vanish on all sites $<i$, so earlier sites are unaffected. Let the updated rows span $L^{(i+1)}$.

\emph{Case $d_i=2$.} Choose $r_{i,1},r_{i,2}\in L^{(i)}$ such that $\{\pi_i(r_{i,1}),\pi_i(r_{i,2})\}$ is a basis of $\pi_i(L^{(i)})$.
For each $u\in L^{(i)}$, replace $u\mapsto u+\alpha(u)r_{i,1}+\beta(u)r_{i,2}$ with the unique $\alpha(u),\beta(u)\in\{0,1\}$ making $\pi_i(u+\alpha r_{i,1}+\beta r_{i,2})=0$.
Again earlier sites are not affected. Let the updated rows span $L^{(i+1)}$.

Iterating over $i=1,\dots,n$, we obtain a basis in ``block'' row echelon form
\[
\mathcal B=\bigsqcup_{i=1}^n \mathcal B^{(i)},\qquad
\mathcal B^{(i)}:=\{\,r_{i,\ell}:\ 1\le \ell\le d_i\,\},
\]
where vectors in $\mathcal B^{(i)}$ are the pivots for site $i$ and $|\mathcal B^{(i)}|=d_i$.

\smallskip
For the second round, fix the following notation from the outcome of Round I:
\[
S_i:=\mathrm{span}\!\Bigl(\bigcup_{m<i}\mathcal B^{(m)}\Bigr),\qquad
T_i:=\mathrm{span}\!\Bigl(\bigcup_{m\ge i}\mathcal B^{(m)}\Bigr),
\]
so that $L=S_i\oplus T_i$ and $T_i=\mathrm{span}(\mathcal B^{(i)})\oplus T_{i+1}$.
We also set the (global) site dimension indicator
\[
k_i:=\dim\bigl(\pi_i(L)\bigr)\in\{0,1,2\}.
\]

\paragraph{Round II}
Proceed for $i=n,\dots,1$ in this order. At step $i$, we keep every vector in $T_i$ fixed, and we are allowed to replace vectors $s\in S_i$ by
\[
s\ \longmapsto\ s'\ :=\ s+\sum_{u\in\mathcal B^{(i)}}\alpha_u(s)\,u\quad (\alpha_u(s)\in\{0,1\}),
\]
i.e., we only add linear combinations of the pivot block $\mathcal B^{(i)}$ to vectors in $S_i$.
This keeps $T_i$ unchanged and cannot alter the values of $s$ at any site $<i$, because each $u\in\mathcal B^{(i)}$ has zero projection on all sites $<i$ (by the echelon property from Round I).

\emph{Case $k_i\in\{0,1\}$.} Do nothing and pass to step $i+1$.

\emph{Case $k_i=2$ and $d_i=2$.} Then $\pi_i(\mathcal B^{(i)})$ already spans $\pi_i(L)$.
For any $s\in S_i$ choose $\alpha_u(s)$ so that $\pi_i\!\bigl(s'\bigr)=0$.
Thus only the two pivot vectors in $\mathcal B^{(i)}$ may remain nonzero at site $i$. Hence two nontrivial Paulis occur there.

\emph{Case $k_i=2$ and $d_i=1$.} Let $v:=\pi_i(r_{i,1})\neq 0$ be the unique pivot direction on site $i$.
Then $\pi_i(L)=\langle v\rangle \oplus (a+\langle v\rangle)$ for any $a\in\pi_i(L)\setminus\langle v\rangle$.
Fix one representative $w\in \pi_i(L)\setminus\langle v\rangle$.
For each $s\in S_i$,
if $\pi_i(s)\in\langle v\rangle$, choose $\alpha_{r_{i,1}}(s)$ to make $\pi_i(s')=0$;
if $\pi_i(s)\in w+\langle v\rangle$, choose $\alpha_{r_{i,1}}(s)$ so that $\pi_i(s')=w$ (the two elements of the affine line $w+\langle v\rangle$ differ by $v$).
Thus, after the replacements, every vector in $S_i$ has projection either $0$ or the single fixed vector $w$ on site $i$, while the only other nonzero projection on site $i$ among the basis vectors is $v$ coming from the pivot row $r_{i,1}$. Hence two nontrivial Paulis occur on site $i$.

\emph{Case $k_i=2$ and $d_i=0$.} In this case all nonzero projections on site $i$ come from $S_i$.
Choose $s^{(1)},s^{(2)}\in S_i$ such that $\{\pi_i(s^{(1)}),\pi_i(s^{(2)})\}$ is a basis of $\pi_i(L)$.
Keep them as part of the basis; for every other $s\in S_i$, replace
\[
s\ \longmapsto\ s+\alpha\,s^{(1)}+\beta\,s^{(2)}\quad(\alpha,\beta\in\{0,1\})
\]
to make the $\pi_i$-projection zero.
Since these operations occur entirely inside $S_i$, $T_i$ remains fixed.
Therefore on site $i$ only the two chosen vectors can be nonzero, and only two nontrivial Paulis occur.

Executing Round II for $i=n,\dots,1$ yields a new basis of $L$ and on each site $i$ at most two distinct nontrivial Pauli operators appear among the basis vectors. This proves the lemma.
\end{proof}

\begin{lem}\label{lem:full-rank}
Let $f_i : \mathbb{F}_2^{\,d} \to \mathbb{F}_2$ be affine functions of the form
\[
    f_i(\mathbf{x}) \;=\; \mathbf{u}_i \cdot \mathbf{x} \,\oplus\, v_i ,
    \qquad i=1,\dots,m,
\]
and define the affine map
\[
    \phi : \mathbb{F}_2^{\,d} \to \mathbb{F}_2^{\,m},
    \qquad
    \phi(\mathbf{x}) = \bigl(f_1(\mathbf{x}),\dots,f_m(\mathbf{x})\bigr).
\]
If $\phi$ is injective, then the coefficient vectors
\[
    \mathbf{u}_1,\dots,\mathbf{u}_m \in \mathbb{F}_2^{\,d}
\]
span the entire space $\mathbb{F}_2^{\,d}$.
\end{lem}

\begin{proof}
Let $U = (\mathbf{u}_1\,|\,\cdots\,|\,\mathbf{u}_m)$ be the $d\times m$ matrix of linear
coefficients. Then
\[
    \phi(\mathbf{x}) = \mathbf{x}^{\mathsf{T}}U \,\oplus\, \mathbf{v}
\]
for a fixed $\mathbf{v}\in\mathbb{F}_2^{\,m}$.
Injectivity of $\phi$ implies injectivity of the linear map
$\mathbf{x}\mapsto \mathbf{x}^{\mathsf{T}}U$, so $U$ has trivial kernel and $\mathrm{rank}(U)=d$.
Thus the columns $\mathbf{u}_1,\dots,\mathbf{u}_m$ span $\mathbb{F}_2^{\,d}$.
\end{proof}

\begin{lem}
[Proposition 2 in Ref.~\cite{carlet2008recursive}]\label{lem:recursive-nl-bound1}
Let $f$ be any $r$-variable Boolean function, and let $s$ be a positive integer smaller than $r$. Then
\begin{equation}\label{eq:recursive-nl}
    \mathrm{nl}_s(f) \;\ge\; \tfrac{1}{2} \max_{\mathbf{a} \in \mathbb{F}_2^r} \mathrm{nl}_{s-1}\!\big(D_{\mathbf{a}} f\big) \, .
\end{equation}
\end{lem}

\begin{lem}[Proposition 3 in Ref.~\cite{carlet2008recursive}]
    \label{lem:recursive-nl-bound2}
Let $f$ be any $r$-variable function and $s$ a positive integer smaller than $r$. We have
\[
\mathrm{nl}_s(f) \;\ge\; 2^{\,r-1} \;-\; \frac{1}{2}\sqrt{\,2^{2r} \;-\; 2 \sum_{\mathbf{a}\in \mathbb{F}_2^r} \mathrm{nl}_{s-1}\!\big(D_{\mathbf{a}} f\big)\,}\,.
\]

\end{lem}

\begin{prop}[Bound via $\mathrm{nl}_2$]\label{prop:bound via derivative}
Let $(\mathcal{C},S)$ be a stabilizer testing game with parity function $c$. If the game admits non-zero quantum advantage, then its classical value satisfies the following two bounds:
\begin{subequations}\label{eq:derivative-bounds}
    \begin{align}
    \text{(i)}\quad
    p_{\mathrm{cl}}^*
    &\le
    \tfrac{3}{4} 
    + 2^{-\tfrac{1}{2}\,\max_{\mathbf{a} \in \mathbb{F}_2^r}\,\mathrm{rank}(D_{\mathbf{a}} c)-2},
    \tag{\ref{eq:derivative-bound1}}\\[4pt]
    \text{(ii)}\quad
    p_{\mathrm{cl}}^*
    &\le
    \tfrac{1}{2}
    + \tfrac{1}{2}\,
    \sqrt{\tfrac{1}{2^r}\sum_{\mathbf{a} \in \mathbb{F}_2^r}
    \!2^{-\tfrac{1}{2}\,\mathrm{rank}(D_{\mathbf{a}} c)}},
    \tag{\ref{eq:derivative-bound2}}
\end{align}
\end{subequations}
where $D_{\mathbf{a}} c(\mathbf{x}) = c(\mathbf{x}) \oplus c(\mathbf{x}+\mathbf{a})$ denotes the Boolean derivative of $c$ in direction $\mathbf{a} \in \mathbb{F}_2^r$, and $\mathrm{rank}(D_{\mathbf{a}} c)$ denotes the rank of the associated quadratic form. 

\end{prop}

\begin{proof}
    By Lemma~\ref{lem:recursive-nl-bound1} and Lemma~\ref{lem:recursive-nl-bound2}, taking $r=2$, we obtain the following two bounds:
\begin{subequations}\label{eq:nl2-bounds}
\begin{align}
    \text{(i)}\quad
    \mathrm{nl}_2(c) 
    &\;\ge\; 
    \tfrac{1}{2}\, \max_{\mathbf{a} \in \mathbb{F}_2^r} 
    \mathrm{nl}_1\!\big(D_{\mathbf{a}} c\big) , 
    \label{eq:nl2-bound-a}\\[4pt]
    \text{(ii)}\quad
    \mathrm{nl}_2(c) 
    &\;\ge\;
    2^{\,r-1}
    \;-\;
    \tfrac{1}{2}
    \sqrt{\,2^{2r} - 2 \!\sum_{\mathbf{a} \in \mathbb{F}_2^r} 
    \mathrm{nl}_1\!\big(D_{\mathbf{a}} c\big)} \, .
    \label{eq:nl2-bound-b}
\end{align}
\end{subequations}
  
    By Eq.~\eqref{eq:nl1-formula} and Eq.~\eqref{eq:walsh-quadratic}, the first-order nonlinearity of $D_{\mathbf{a}} c$ can be written as
    \begin{equation}\label{eq:nl1-derivative}
        \mathrm{nl}_{1}(D_{\mathbf{a}} c) \;=\; 2^{\,r-1} \;-\; 2^{\,r-1 - \tfrac{1}{2}\,\mathrm{rank}(D_{\mathbf{a}} c)} \, .
    \end{equation}

    Substituting \eqref{eq:nl1-derivative} into the bound \eqref{eq:nl2-bounds}, and combining with Eq.~\eqref{eq:bound-by-nl2}, we arrive at Eq.~\eqref{eq:derivative-bounds}.

\end{proof}
\section{Coding-theory bounds}\label{app:examples}
\resub{In this Appendix, we illustrate Proposition \ref{prop:bound via derivative} with two examples. First, using Eq.~\eqref{eq:derivative-bound1}, we derive an upper bound on the classical value of the $(\mathcal{C}_{\mathrm{toric}}(L), S_{\mathrm{toric}})$ stabilizer game, with $S_{\mathrm{toric}}$ the full stabilizer group of the 2D toric code on an $L\times L$ torus, demonstrating that its asymptotic difference is at least $1/2$.}

\begin{figure}[t]
    \centering
    \includegraphics[width=\linewidth]{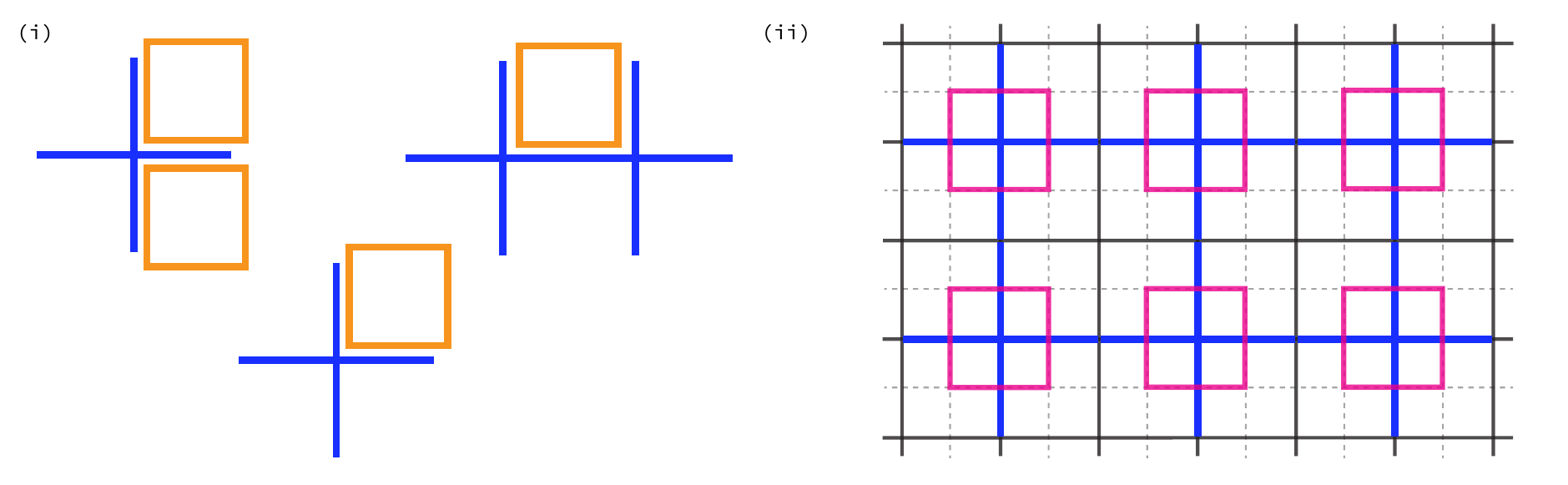}
        
    \caption{Visualizing the adjacency graph of the quadratic form obtained from the Boolean derivative of the toric-code parity function.
    (i) Schematic illustration of the three types of terms contributing to the derivatives in Eqs.~\eqref{eq:derivative_cycle_0} and~\eqref{eq:derivative_cycle}, following the visual conventions of Fig. \ref{fig:toric_code_stabilizers}. Overlaps between triplets of distinct stabilizers give rise to quadratic terms while overlaps between pairs of distinct stabilizers give rise to linear terms in these expressions. (ii) Choice of a Boolean derivative direction $\mathbf{a}$ that couples to every other star. The stars in the support of $\mathbf{a}$ form a set of $\lfloor L/2\rfloor^2$ equally spaced, pairwise non-adjacent vertices of the direct lattice, leading to plaquette--plaquette interactions around each star, illustrated in the figure by pink bonds on the dual lattice. The plaquette--plaquette block of the resulting quadratic form has rank at least $2\lfloor L/2\rfloor^2$, leading to the bound in Eq.~(57).}
    \label{fig:toric_code}
\end{figure}
\begin{prop}[Toric code]\label{thm:toric-bound}
    Let $(\mathcal{C}_{\mathrm{toric}}(L), S_{\mathrm{toric}})$ be a stabilizer testing game
    constructed from the two-dimensional toric code on an $L\times L$ torus, with $n=2L^2$ qubits. 
    Then the classical value satisfies
    \begin{equation}
        \label{eq:toric-bound}
        p_{\mathrm{cl}}^*(\mathrm{toric}) \leq \frac{3}{4}+2^{-\lfloor L/2 \rfloor^{2}-2}.
    \end{equation}
\end{prop}

\begin{proof}

Let $\{\mathbf{e}_i\}_{i=1}^n$ denote the standard basis of $\mathbb{F}_2^n$. We compute derivatives along these basis directions. We label the star generators with indices in $A = \{1,\dots,L^2\}$ and the plaquette generators with indices in $B = \{L^2+1,\dots,2L^2\}$.  
By Eq.~\eqref{eq:parity_toric} and using analogous notation, the derivative along the basis element $\mathbf{e}_i$ with $i \in A$ gives rise to
four plaquette-plaquette terms, eight star-plaquette terms and a linear sum over four plaquettes:
\begin{equation}\label{eq:derivative_cycle_0}
D_{\mathbf{e}_i}c(\mathbf{x}) \;=\; \sum_{\{\{p,p'\} : s_i \cap p \cap p' \neq \emptyset\}} x_p x_{p'} 
\;+\; \sum_{\{\{p,s\}: s_i \cap s \cap p \neq \emptyset\}} x_s x_p + \sum_{\{p : s_i \cap p \neq \emptyset\}} x_p. 
\end{equation}


Similarly, taking the derivative along the basis element $\mathbf{e}_i$ with $i \in B$ gives rise to four star-star terms, eight star-plaquette terms and a linear sum over four stars:
\begin{equation}
\label{eq:derivative_cycle}
D_{\mathbf{e}_i} c(\mathbf{x})
= \sum_{ \{ \{s,s'\} : p_{i-L^2} \cap s \cap s' \neq \emptyset\}} x_s x_{s'} + \sum_{\{\{s,p\}:\\ p_{i - L^2} \cap s \cap p \neq \emptyset\}} x_s x_p + \sum_{\{s : p_{i-L^2}\cap s \neq \emptyset\}} x_s.
\end{equation}
Fig.~\ref{fig:toric_code} illustrates schematically how each of these terms arises.

As shown in Fig.~\ref{fig:toric_code}, we now select a direction $\mathbf{a}$ for the derivative that couples to
every other star. The $\lfloor L/2 \rfloor^{2}$ stars in the support of $\mathbf{a}$ are thus equally spaced and pairwise non-adjacent. Note that the rank of the quadratic form associated with $D_{\mathbf{a}} c$ is bounded below by the rank of the principal submatrix consisting of plaquette-plaquette terms. Each such term contributes rank $2$ to this submatrix. Therefore, the rank of the derivative of the parity function along this direction is at least $2\lfloor L/2 \rfloor^{2}$. Applying Eq.~\eqref{eq:derivative-bound1}, we deduce Eq.~\eqref{eq:toric-bound}.
\end{proof}
        
\resub{Next, using Eq.~\eqref{eq:derivative-bound2}, we derive an upper bound on the classical value of the stabilizer-testing game for the cyclic cluster state.

\begin{prop}[Cyclic cluster state]
    Let $(\mathcal{C}_H,S_H)$ be the stabilizer-testing game constructed from the cyclic cluster state $|C_n\rangle$ on the $n$-cycle $H$ and $n>4$. Its classical value satisfies 
    \begin{equation}
        p_{\mathrm{cl}}^* \leq \tfrac{1}{2}
+ \tfrac{1}{2}
\left( \cos{\frac{\pi}{8}} \right)^{\lceil \frac{n}{2} \rceil-2}.
    \end{equation}
\end{prop}
\begin{proof}
    Write the vertex set of the $n$-cycle $H$ as $V=V_1\cup V_2$, where $V_1$ and $V_2$ correspond to odd and even sites.   
    To apply bound (ii) in Proposition 3 to the cyclic cluster state, 
    we again bound $\operatorname{rank}(D_{\mathbf{a}}c)$ by the rank of a principal submatrix of $D_{\mathbf{a}}c$ as in the proof of Prop. \ref{thm:toric-bound}. 
    Specifically we consider the principal submatrix induced by the subset of even sites 
    \begin{equation}
        V_2' = \begin{cases} V_2 \setminus \{n\}, & n \, \mathrm{even}, \\ V_2, & n \, \mathrm{odd}. \end{cases}
    \end{equation}
    Focusing without loss of generality~\cite{macwilliams1977theory} on the quadratic part of $D_{\mathbf{a}}c$ restricted to this subspace then yields the quadratic form
    \begin{equation}
    Q_2(\mathbf{x}) = \begin{cases} a_3x_2x_4 + a_5 x_4x_6 + ... + a_{n-3} x_{n-4}x_{n-2}, & n \, \mathrm{even}, \\
    a_3 x_2 x_4 + a_5x_4x_6 + ... + a_{n-2} x_{n-3}x_{n-1}, & n \, \mathrm{odd}.
    \end{cases}
    \end{equation}
Let $V_1' \subset V_1$ denote the subset of odd sites appearing in these expressions and let $l =|V'_1|=\lceil n/2\rceil-2 \geq 1$. Also let  $\mathbf{a}_1\coloneqq \mathbf{a}|_{V'_1}$ denote the restriction of $\mathbf{a}$ to $V'_1$ and $|\mathbf{a}_1|$ its Hamming weight. It is immediate that
\begin{equation}
\operatorname{rank}(D_{\mathbf a}c)
\;\ge\;
\operatorname{rank}\!\left(\left.(D_{\mathbf a}c)\right|_{V'_2\times V'_2}\right) = \operatorname{rank}\!\left(Q_2(\mathbf{x})\right)
= |\mathbf a_1|.
\end{equation}
In Eq. (57b), this yields 
\begin{align}
2p_{\mathrm{cl}}^* - 1 \leq \sqrt{\frac{1}{2^n}
\sum_{\mathbf{a} \in \mathbb{F}_2^n}
    \!2^{-\tfrac{1}{2}\,\mathrm{rank}(D_{\mathbf{a}} c)}} 
\leq \sqrt{\frac{1}{2^n} \cdot 2^{n-l} \sum_{\mathbf{a_1} \in \mathbb{F}_2^l}
     2^{-\frac{|\mathbf{a}_1|}{2}}} 
= \sqrt{\frac{1}{2^l} \sum_{k=0}^l \binom{l}{k} 2^{-\frac{k}{2}}} 
 = 
\left(\cos{\frac{\pi}{8}} \right)^{l} 
\end{align}
and the result follows.
\end{proof}}

\section{Exact classical values for cyclic cluster states}
\label{app:numerics}
In this Appendix, we tabulate the exact classical values of the stabilizer-testing game for cyclic cluster states, extending previous results for $n \leq 10$~\cite{Cabello_2013,Daniel_2022} to all values of $3 \leq n \leq 16$ by an exhaustive (and therefore exact) numerical maximization of Eq. \eqref{eq:pclstarcluster}. We record the value of the lower bound Eq. \eqref{eq:pcllowerbound}, the exact value of $p_{\mathrm{cl}}^*$ and the threshold $\mathcal{F}_c = 2p_{\mathrm{cl}}^*-1$ such that a fidelity $\mathcal{F} > \mathcal{F}_c$ to the cyclic cluster state witnesses contextuality by Eq. \eqref{eq:fidelityformula}. 
See Table. \ref{tab:cycliccluster}.
\begin{table}[h]
    \centering
    \begin{tabular}{c|c|c|c}
         $n$ & Lower bound from Eq. \eqref{eq:pcllowerbound} & Exact value of $p_{\mathrm{cl}}^*$ & $\mathcal{F}_c$ \\
         \hline
 $3$ & $\frac{7}{8}$ & $\frac{7}{8}$ & $0.7500$ \\
$4$ & $\frac{12}{16}$ & $\frac{14}{16}$ & $0.7500$ \\
$5$ & $\frac{26}{32}$ & $\frac{26}{32}$ & $0.6250$ \\
$6$ & $\frac{46}{64}$ & $\frac{46}{64}$ & $0.4375$ \\
$7$ & $\frac{92}{128}$ & $\frac{92}{128}$ & $0.4375$ \\
$8$ & $\frac{176}{256}$ & $\frac{176}{256}$ & $0.3750$ \\
$9$ & $\frac{340}{512}$ & $\frac{340}{512}$ & $0.3281$ \\
$10$ & $\frac{664}{1024}$ & $\frac{664}{1024}$ & $0.2969$ \\
$11$ & $\frac{1288}{2048}$ & $\frac{1288}{2048}$ & $0.2578$ \\
$12$ & $\frac{2520}{4096}$ & $\frac{2520}{4096}$ & $0.2305$ \\
$13$ & $\frac{4928}{8192}$ & $\frac{4928}{8192}$ & $0.2031$ \\
$14$ & $\frac{9664}{16384}$ & $\frac{9664}{16384}$ & $0.1797$\\
$15$ & $\frac{18992}{32768}$ & $\frac{18992}{32768}$ & $0.1592$\\
$16$ & $\frac{37376}{65536}$ & $\frac{37376}{65536}$ & $0.1406$
    \end{tabular}
    \caption{Exact classical values $p_{\mathrm{cl}}^*$ and corresponding fidelity thresholds $\mathcal{F}_c$ to four decimal places for the stabilizer-testing game played with cyclic cluster states.}
    \label{tab:cycliccluster}
\end{table}

\end{document}